\documentclass[12pt]{article}
\usepackage[english]{babel}

\usepackage[T1]{fontenc}
\usepackage[nottoc,notlot,notlof]{tocbibind}
\usepackage{amssymb,mathrsfs,amsfonts}
\usepackage{amsmath, amsthm, xpatch, dsfont,ushort,sgamevar, float, framed, multirow, array, makecell, comment, url, graphicx, subcaption, mathtools, fnpct, setspace, pifont, bm, bbm,accents}

\usepackage{nicefrac, xfrac}
\usepackage{epstopdf}
\usepackage[normalem]{ulem}

\usepackage{lmodern}
\usepackage[final, babel, activate={true,nocompatibility}]{microtype}
\DeclarePairedDelimiterX{\infdivx}[2]{(}{)}{%
#1\;\delimsize\|\;#2%
}

\usepackage{lipsum}
\usepackage[utf8]{inputenc} 
\usepackage[T1]{fontenc} 
\usepackage{titlesec}
\usepackage{thmtools}

\usepackage[round,comma,authoryear,longnamesfirst,sort&compress]{natbib}
\usepackage[bottom, multiple]{footmisc}
\usepackage[top=1.5in, bottom=1.5in, left=1.25in, right=1.25in]{geometry}
\usepackage{xcolor}
\definecolor{navyblue}{rgb}{0.0, 0.13, 0.5}
\definecolor{amaranth}{rgb}{0.9, 0.17, 0.31}
\definecolor{brightpink}{rgb}{1.0, 0.0, 0.5}
\definecolor{cgreen}{rgb}{0.0, 0.42, 0.24}
\definecolor{cobalt}{rgb}{0.0, 0.28, 0.67}
\definecolor{coquelicot}{rgb}{1.0, 0.22, 0.0}
\definecolor{deepcarrotorange}{rgb}{0.91, 0.41, 0.17}
\definecolor{portlandorange}{rgb}{1.0, 0.35, 0.21}
\definecolor{burntorange}{rgb}{0.8, 0.33, 0.0}
\definecolor{dukeblue}{rgb}{0.0, 0.0, 0.61}
\definecolor{brickred}{rgb}{0.8, 0.25, 0.33}
\definecolor{brown(web)}{rgb}{0.65, 0.16, 0.16}
\definecolor{ceruleanblue}{rgb}{0.16, 0.32, 0.75}

\usepackage{hyperref}
\hypersetup{
colorlinks=true,
linkcolor=black,
filecolor=black,
citecolor=black,
urlcolor=dukeblue,
}
\usepackage[capitalize,noabbrev,nameinlink]{cleveref}

\DeclareMathAlphabet{\mathsf}{OT1}{\sfdefault}{m}{n}
\SetMathAlphabet{\mathsf}{bold}{OT1}{\sfdefault}{b}{n}

\usepackage{enumitem}

\newcommand{\argmax}{\arg\!\max}

\renewcommand\thmcontinues[1]{Continued}

\newtheorem{claim*}{Claim}[]

\newcommand*\diff{\mathop{}\!\mathrm{d}}
\usepackage{tocloft}

\newtheorem{prop*}{Proposition}[]
\newtheorem{bprop*}{Proposition}
\newtheorem{bthm*}{Theorem}
\newtheorem{thm*}{Theorem}[]

\newtheorem{ax*}{Axiom}[]
\newtheorem{lem*}{Lemma}[]
\newtheorem{asp*}{Assumption}[]
\newtheorem{cor*}{Corollary}[]
\newtheorem{pty*}{Property}[]
\newtheorem{def*}{Definition}[]
\newtheorem{cond*}{Condition}[]

\newtheorem*{propnonumber*}{Proposition}

\usepackage{caption}
\usepackage[labelfont=bf]{caption}

\usepackage{tikz}
\usepackage{pgfplots}
\usetikzlibrary{automata,positioning}
\usetikzlibrary{arrows,intersections}

\newtheorem{examplex}{Example}
\newenvironment{eg*}
{\pushQED{\qed}\examplex}
{\popQED\endexamplex}

\newtheorem{remarkx}{Remark}
\newenvironment{rem*}
{\pushQED{\qed}\remarkx}
{\popQED\endremarkx}

\newcommand{\ra}{\rightarrow}

\newcommand{\ve}{\varepsilon}

\newcommand{\bE}{\mathbf{E}}
\newcommand{\bR}{\mathbf{R}}

\newcommand{\expo}{ \operatorname{e}}
\newcommand{\gl}{\lambda}
\newcommand{\ga}{\alpha}
\newcommand{\gb}{\beta}
\newcommand{\gd}{\delta}
\newcommand{\gs}{\sigma}
\newcommand{\Gs}{\Sigma}
\newcommand{\gam}{\gamma}

\newcommand{\egg}{e.g.,~}

\newcommand{\BR}{\mathrm{BR}}

\newcommand{\KL}[2]{D(#1  \parallel #2 )}

\usepackage[toc]{appendix}

\renewcommand{\bibname}{References}

\crefformat{onlineappendix}{#2Supplementary Appendix~#1#3}
\Crefformat{onlineappendix}{#2Supplementary Appendix~#1#3}

\renewcommand{\texttt}[1]{%
\begingroup
\ttfamily
\begingroup\lccode`~=`/\lowercase{\endgroup\def~}{/\discretionary{}{}{}}%
\begingroup\lccode`~=`[\lowercase{\endgroup\def~}{[\discretionary{}{}{}}%
\begingroup\lccode`~=`.\lowercase{\endgroup\def~}{.\discretionary{}{}{}}%
\catcode`/=\active\catcode`[=\active\catcode`.=\active
\scantokens{#1\noexpand}%
\endgroup
}

\makeatletter
\renewenvironment{proof}[1]
[\proofname]{\par\pushQED{\qed}\normalfont%
\topsep6\p@\@plus6\p@\relax
\trivlist\item[\hskip\labelsep\bfseries#1\@addpunct{.}]%
\ignorespaces}{%
\popQED\endtrivlist\@endpefalse
}
\makeatother

\pgfplotsset{compat=1.14}
\crefname{prop*}{Proposition}{Propositions}
\Crefname{prop*}{Proposition}{Propositions}
\usetikzlibrary{calc}
\usetikzlibrary{matrix}
\usetikzlibrary{positioning}
\usetikzlibrary{decorations.pathreplacing}
\allowdisplaybreaks

\usepackage{etoolbox}
\patchcmd{\quote}{\rightmargin}{\leftmargin 4em \rightmargin}{}{}

\makeatletter
\AddToHook{cmd/appendix/before}{\def\cref@section@alias{appendix}\def\cref@subsection@alias{appendix}}
\makeatother

\title{The Empirical Content of Reputation Effects\thanks{I thank Navin Kartik, \href{https://refine.ink/}{Refine.ink}, Harry Pei, Larry Samuelson, Satoru Takahashi, and audiences at various seminars for helpful comments.}}
\author{Allen Vong\thanks{National University of Singapore, \href{mailto:allenv@nus.edu.sg}{\texttt{allenv@nus.edu.sg}}.}}
\date{October 5, 2026}

\begin{document}
\maketitle

\begin{abstract}
I revisit the classic reputation model in which a long-lived player, either normal or committed, faces short-lived players who may misspecify the commitment-type signal process. I show that the long-lived player's patience makes reputation powerful but leaves its value sensitive both to short-lived players' subjective commitment-type signal process and to discounting. Every subjective process is arbitrarily close over any fixed horizon to processes eliminating his reputation gains and also to processes preserving them. Moreover, generically, all equilibria yield substantial gains along some discount-factor sequences approaching one, while even the best equilibrium yields no positive gain in the limit along others.
\smallskip

\end{abstract}

\thispagestyle{empty}
\pagebreak
\pagenumbering{arabic}

\section{Introduction}

A reputation for commitment can be valuable. A firm thought to be committed to quality, for example, can charge markups. The insight that such a reputation can have large equilibrium effects, pioneered by \citet{kreps1982rational}, \citet{kreps1982reputation}, and \citet{milgrom1982predation}, is central to game theory. 

\citet{fudenberg1989,fudenberg1992} develop the canonical reputation framework. A long-lived player is privately either a normal type who acts strategically or a commitment type who plays a fixed action (distribution) and faces a sequence of short-lived players. In every Nash equilibrium, a sufficiently patient normal type secures a payoff arbitrarily close to what he would obtain if the short-lived players treated him as the commitment type; he does so by deviating to persistent play of the commitment type's action. This payoff guarantee can strictly exceed his highest complete-information equilibrium payoff; this is the ``reputation effect.'' This effect requires that the long-lived player is sufficiently patient, so that his future gains justify the initial costs of influencing the short-lived players' beliefs through the deviation.

The canonical framework imposes the substantive assumption that a commitment type and a normal type persistently playing the commitment action generate the same signal process. Because commitment is latent and rare, the short-lived players may have little basis for interpreting the signals it generates and hold commitment to an incorrect standard. Consumers, for example, may incorrectly believe that a firm committed to high effort adopts protocols that gradually improve quality provision, and so regard noisy bad signals later in this firm's history as far less likely than they actually are. Do reputation effects survive when short-lived players slightly misspecify the commitment type's signal process relative to the canonical framework?\footnote{The canonical framework's primitives treat the commitment type as an objective latent type, rather than as the short-lived players' subjective hypothesis. Following the misspecification literature, I call departures from the canonical commitment-type signal process misspecifications.} 

I extend the canonical framework to allow for this misspecification. When is a misspecification slight? I adopt a criterion with an empirical interpretation, comparing signal processes by what can be assessed from observables. As players observe only a finite history at any time, I say that misspecification vanishes empirically along a sequence of the short-lived players' subjective commitment-type signal processes if, for every fixed finite horizon, the sequence of distributions of signal histories they induce converges to the true one. I call the topology this convergence induces the finite-history topology and view a misspecification as slight if it lies in a small neighborhood of the true process in this topology. A more permissive criterion would treat misspecification as slight despite persistent statistical differences at some finite horizon; a stricter one would exclude some sequences of subjective processes even though these differences vanish at every fixed finite horizon. My criterion follows \citet{weinstein2007structure}, who also study the robustness of strategic predictions---in their case, rationalizability---to perturbations that can be assessed from observables---in their case,  belief hierarchies disciplined by finitely many orders of beliefs.

My results show that when the reputation effect is present in the canonical framework and short-lived players may misspecify the commitment-type signal process, however slightly, equilibrium payoffs can remain sensitive both to how these players interpret long histories and to the discount factor, however close it is to one. Under the finite-history topology, every neighborhood of every subjective process of the short-lived players, including the correctly specified one, contains a process under which the reputation effect collapses---the patient upper limit of the normal type's highest equilibrium payoff is at most the complete-information benchmark (\cref{thm:fd-collapse})---and another under which it survives (\cref{thm:nowherecontinuity}). Moreover, for a generic subjective process, the long-lived player's sufficient patience alone does not settle the payoff prediction: all equilibria yield substantial reputation gains along some discount-factor sequences approaching one, while even the best equilibrium yields no positive reputation gain in the limit along others (\cref{thm:generic-instability}). Consequently, while the density results leave open whether either the collapse of the reputation effect or its survival is typical, \cref{thm:generic-instability} shows that neither is.

These results hold because finite-horizon proximity imposes little restriction on the likelihoods of long histories. The same patience that makes the reputation effect powerful in the canonical framework makes those long histories payoff-relevant, allowing misspecification to vanish over every fixed finite horizon yet have large payoff consequences. My results therefore show that the canonical reputation prediction rests on substantive assumptions about how the audience interprets long histories. The finite-history topology makes their economic role explicit: even arbitrarily close agreement over every fixed horizon does not determine whether reputation gains survive or whether sufficient patience yields a stable payoff prediction. The assumptions needed to sustain the canonical guarantee require justification beyond any fixed amount of finite-history evidence about the signal process.


\paragraph{Related literature.} \citet{jehiel2012reputation} first introduce misspecification into the canonical reputation framework: short-lived players interpret signals through analogy-based models. My model does not concern a particular behavioral model of inference; players here are Bayesian. My departure from the canonical setup is the short-lived players' subjective specification of the commitment-type signal process. Relatedly, \citet{ely2026ruth} studies reputation and misspecification, but in a different framework: all players are long-lived, there is no commitment type, misspecification arises from one player misperceiving her signal structure, and reputation concerns arise from the other players' beliefs about that signal structure.

Recent work introduces identification and signaling into the canonical reputation framework. In \citet{pei2020reputation}, the long-lived player is privately informed about not only his type but also a payoff-relevant state. In \citet{deb2025reputation}, the same signal distribution may arise from commitment behavior in one state and from strategic behavior warranting a different short-lived-player response in another. In \citet{luo2025marginal}, the long-lived player observes private signals before taking public actions, so the marginal distribution of actions does not identify his strategy. In \citet{pei2024reputation, pei2026reputationeffects}, the short-lived players observe only a coarsened signal history. My model instead retains the canonical structure and departs by introducing misspecification of the commitment-type signal process.

\citet{sorin1999merging} applies the merging-of-opinions approach of \citet{kalai1993rational} to establish the reputation guarantee using a grain-of-truth condition satisfied in the canonical framework: the short-lived players' perceived signal process is a mixture that puts positive weight on the true commitment-type signal process. My results show that this condition remains sufficient for the guarantee under misspecification, while its failure can allow all reputation rents to disappear.

In proving my survival result, I adapt the entropy techniques introduced to the reputation literature by \citet{gossner2011simple}, who derives explicit payoff bounds under correct specification, improving on the implicit bounds of \citet{fudenberg1992}; \citet{ekmekci2012impermanent} and \citet{faingold2020reputation} exploit this improvement in their joint-limit analyses under correct specification. I use entropy to instead control a mismatch due to misspecification: the short-lived players' equilibrium predictions depend on their subjective commitment-type signal process, but signals are generated by the true one. This mismatch prevents a direct extension of the active-supermartingale argument of \citet{fudenberg1992} to my setting, because the short-lived players' misspecified posterior odds of the normal type relative to the commitment type need not be a supermartingale.\footnote{Relatedly, in a single-agent sequential decision problem, \citet{gossner2008entropy} bound the decision maker's payoff loss in terms of the relative entropy between the true and subjective signal distributions.}


My results are related to \citet{cavounidis2025blackwell}, who study strategy profiles that remain equilibria for all high enough discount factors in repeated games. I instead study the value of reputation as patience increases, allowing equilibrium strategies to vary with discounting.

The misspecification literature largely focuses on individual and social learning and asks whether asymptotic outcomes are robust to misspecified inferences (see \citealp{bohren2025misspecified} for a survey). For individual learning, \citet{berk1966limiting} establishes robustness. For social learning, robustness holds in some settings (\egg \citealp{bohren2021learning}) but not in others (\egg \citealp{frick2020misinterpreting}).

\section{Motivating example}
\label{sec:productchoice}

In this section, I illustrate my main results in the product-choice example of \citet{cripps2004imperfect}. A long-lived firm with discount factor $\gd \in (0,1)$ interacts with a sequence of short-lived consumers, each of whom lives for one period. In each period, the firm and the consumer simultaneously choose actions in the game
\begin{align}\label{eq:matrix}
\begin{array}{c|cc}
& b_h & b_l \\
\hline
a_h & 2,3 & 0,2 \\
a_l & 3,0 & 1,1
\end{array}
\end{align}
with the row player being the firm and the column player being the consumer. The firm chooses whether to exert high effort, $a_h$, or low effort, $a_l$. The consumer chooses whether to buy a high-priced product, $b_h$, or a low-priced product, $b_l$. This stage game has a unique and strict Nash equilibrium, $(a_l,b_l)$.

The firm's action is hidden from consumers but monitored through a public signal $y \in \{y_h,y_l\}$, where $y_h$ represents a high-quality signal and $y_l$ represents a low-quality signal. Conditional on the firm's action $a$, the signal distribution is
\begin{align}
\Pr(y = y_h | a) = 1 - \Pr(y = y_l | a) =
\begin{cases}
~p, \quad & \text{if } a = a_h, \\
~p', \quad & \text{if } a = a_l,
\end{cases}
\quad 0 < p' < p < 1. \label{eq:examplemonitoring}
\end{align}
Each consumer's action is hidden from future consumers. As is standard, the consumer's payoffs in \eqref{eq:matrix} are ex ante payoffs induced by the signal structure \eqref{eq:examplemonitoring} from realized payoffs given her action and the signal.\footnote{Her realized payoff from action $b_h$ is $3(1-p')/(p-p')$ given signal $y_h$ and $-3p'/(p-p')$ otherwise; her payoff from action $b_l$ is $1 + (1-p')/(p-p')$ given signal $y_h$ and $1-p'/(p-p')$ otherwise.}

Let $u_t$ denote the firm's realized payoff in period $t$. In the repeated game, the firm's realized payoff is $(1-\gd) \sum_{t=0}^{\infty} \gd^t u_t$. If the firm were able to commit to exerting high effort in each period and thereby induce consumers to buy the high-priced product, it would achieve payoff $2$. But it has no commitment power. By standard arguments \citep{abreu1990toward},\footnote{Their methodology applies to perfect public equilibria. The payoff-equivalence argument of \citet[Theorem 5.2]{fudenberg1994efficiency} can be adapted here to show that Nash, perfect public, and sequential equilibria yield the same payoffs. That theorem assumes that the short-lived players' actions are public, but this assumption is unnecessary for the payoff bound here. The payoff-equivalence between perfect public and sequential equilibria follows because there is only one long-lived player, and against consumers whose strategies depend only on public signal histories, he has a best response that also depends only on those histories. The payoff-equivalence between Nash and perfect public equilibria follows because all signal histories are on path.} its Nash equilibrium payoff is at most
\begin{align}
\label{eq:CIupperbound}
\max\!\left( 1, 2 - \frac{1-p}{p-p'} \right)\!,
\end{align}
and this bound is attainable in equilibrium if $\gd$ is sufficiently close to one. This is nonetheless strictly lower than the commitment payoff $2$.

Now suppose there is incomplete information. With probability $\mu_0\in(0,1)$, the firm is a commitment type who always exerts high effort; otherwise, it is a normal type who strategically chooses between high and low effort each period. The firm privately knows its type. \citet{fudenberg1992} show that, for any $w<2$, if $\gd$ is sufficiently close to one, every Nash equilibrium gives the normal-type firm a payoff of at least $w$.\footnote{This example focuses on a pure-action commitment type. More generally, for any $w<5/2$, if the commitment type plays $a_h$ with probability higher than but sufficiently close to $1/2$, then, for $\gd$ sufficiently close to one, the normal type's Nash equilibrium payoff is at least $w$.} This is the reputation effect: even a small probability of commitment guarantees a patient firm a payoff arbitrarily close to the commitment payoff in every Nash equilibrium, exceeding what it can obtain under complete information.

Suppose in addition that consumers misspecify the signal process for the commitment-type firm. Consider two examples:
\begin{enumerate}\itemsep0em
\item Consumers believe that a firm committed to high effort gradually improves its quality-control technology, making signal $y_l$ rarer as it matures. For some $n=1,2,\ldots$, they believe that the commitment type's signals are independent across periods, producing $y_l$ with probability $(1-p)\expo^{-(t/n)^2}$ and signal $y_h$ otherwise in each period $t$. 
\item Consumers believe that a firm committed to high effort has better production technology than the normal type. They believe that the commitment type's signals are i.i.d., producing $y_h$ with probability $p+1/n$ and $y_l$ otherwise in each period, where $n$ is a sufficiently large integer.
\end{enumerate}

Both misspecifications vanish over every fixed finite horizon as $n\ra\infty$: in each period $t$, the commitment-type signal distribution converges to the true one. Yet my results imply that their payoff implications differ. For every fixed $n$ in the first example, the patient upper limit of the firm's highest equilibrium payoff is at most the complete-information benchmark \eqref{eq:CIupperbound}---the reputation effect collapses. In the second, for any $w<2$, sufficiently large $n$ ensures that every equilibrium gives a sufficiently patient firm at least $w$---the reputation effect survives.

What drives the distinction between them? How robust are these predictions? I next develop a general model to address these questions.

\section{Model}
\label{sec:model}

I retain the canonical reputation model and allow only the short-lived players' perceived commitment-type signal process to be misspecified. Time $t=0,1,\ldots$ is discrete and the horizon is infinite. A long-lived player 1 (he), with discount factor $\gd\in(0,1)$, faces a sequence of short-lived players 2 (she), one per period. Player 1's private type is $\xi\in\Xi:=\{\xi^0,\widehat\xi\}$, where $\xi^0$ is the normal type and $\widehat\xi$ the commitment type.\footnote{My results extend to finitely many commitment types. For the collapse result in \cref{thm:fd-collapse}, my construction makes all commitment types sufficiently unlikely at the same time. For the survival result in \cref{thm:nowherecontinuity}, my entropy argument applies to imitating the commitment type with the highest canonical guarantee. The genericity result in \cref{thm:generic-instability} extends by the same semicontinuity and Baire category arguments.} The prior probability of the commitment type is $\mu_0\in(0,1)$.

In period $t$, the players simultaneously choose $a_t\in A$ and $b_t\in B$. Player 1's action is hidden and generates a public signal $y_t\in Y$ according to the distribution $\rho(\cdot\mid a_t)$.\footnote{My results extend to monitoring that depends on both $a_t$ and $b_t$, with an adjusted payoff guarantee based on entropy-confirming best responses, as is standard; see \cref{sec:slp-dependent-monitoring} for details.} Player 2's realized action is observable to player 1 but is hidden from subsequent short-lived players; my main results are unchanged if it is public. The sets $A$, $B$, and $Y$ are finite, and each contains at least two elements. Monitoring is imperfect and has full support: $\rho(y\mid a)>0$ for every $(y,a)\in Y\times A$. Player 1's ex ante stage payoff is $u:A\times B\ra\bR$. Player 2's realized payoff is $\tilde v:B\times Y\ra\bR$, so her ex ante stage payoff is $v(a,b):=\bE_{y\sim\rho(\cdot\mid a)}[\tilde v(b,y)]$. I extend $\rho$, $u$, $\tilde v$, and $v$ to mixed actions in the usual way. For $\ga\in\Delta(A)$, write $\rho_\ga:=\sum_{a\in A}\ga(a)\rho(\cdot\mid a)$.

Let $h_t=(y_0,\ldots,y_{t-1})\in Y^t$ be the public history at time $t$. Player 1's behavior strategy $\gs_t^1:Y^t\times\Xi\ra\Delta(A)$ maps public history and type into a mixed action.\footnote{Restricting the normal type from conditioning on his past private actions is innocuous for my main results. Against short-lived players whose strategies depend only on public history, he has a best response with the same dependence.} The commitment type plays a fixed mixed action $\widehat\ga$ after every history: $\gs_t^1(\cdot,\widehat\xi)=\widehat\ga$. Player 2's strategy $\gs_t^2:Y^t\ra\Delta(B)$ maps public history into a mixed action. Write $\gs^1 := (\gs_t^1)_{t=0}^\infty$, $\gs^2 := (\gs_t^2)_{t=0}^\infty$, and $\gs = (\gs^1, \gs^2)$. Because the commitment type's strategy is fixed, for convenience, I sometimes refer to $\gs^1$ as the normal-type strategy.

Let $H:=Y^\infty$, endowed with the product topology and its Borel sigma-algebra, and let $\Delta(H)$ be the set of probability measures on $H$. Throughout, I identify a signal process with its induced probability measure on $H$; conditional distributions at null histories may be specified arbitrarily. The true signal process conditional on the commitment type is denoted by $\widehat P$; it is i.i.d.\ with one-period marginal $\widehat\rho:=\rho_{\widehat\ga}$. The short-lived players perceive instead a subjective signal process $F\in\Delta(H)$ for the commitment type.\footnote{This specification has a subjective-interdependent-values interpretation: conditional on a given action, the short-lived players believe the signal distribution depends on player 1's type. For a model with objective type-dependent signals, see the chain store game of  \citet{wiseman2009reputation}.} Correct specification is the special case $F=\widehat P$. I impose no independence, Markov, or full-support restriction on $F$. Thus, although $\widehat P$ is i.i.d., $F$ may misspecify either its one-period marginal or its serial dependence. For any $P\in\Delta(H)$ and horizon $T=1,2,\ldots$, let $P_T$ denote the marginal of $P$ on $Y^T$. I say that $F$ has full support if, for every finite horizon $T$, $F_T$ has full support. Throughout, I refer to a subjective commitment-type signal process simply as a subjective process.

The subjective process $F$ is a primitive known to the players.\footnote{The normal type therefore accounts for the short-lived players' misspecification when choosing his strategy in equilibrium. A similar knowledge assumption appears in \citet{frick2020misinterpreting} in which subjective misspecified beliefs are taken as primitives and, in the version of their model where some agents are correctly specified and some are not, correctly specified agents know the misspecified agents' beliefs.} My results extend if the long-lived player does not know $F$, but only knows that it belongs to a finite set $\mathcal F$ and evaluates his payoff in the worst case over that set. See \Cref{sec:robust} for details.\footnote{Briefly, the extension follows because the arguments in my results make little use of the long-lived player's knowledge of $F$: the collapse construction in \cref{thm:fd-collapse} can be applied simultaneously to all candidate processes in $\mathcal F$, the survival argument in \cref{thm:nowherecontinuity} comes from a deviation the normal type can play without knowing $F$, and the genericity result in \cref{thm:generic-instability} extends by the same semicontinuity and Baire category arguments.}

My results also cover short-lived players who hold a prior belief over multiple subjective commitment-type signal processes and update it as histories accumulate. Their Bayesian predictions and best responses are the same as if they used the prior-weighted mixture of these processes as their subjective process from the outset. Since this mixture is itself a subjective process $F \in \Delta(H)$,  taking a single process as the primitive entails no loss of generality.

The short-lived players' beliefs are as follows. Let $L^{\gs^1}$ denote the signal process conditional on the normal type who plays strategy $\gs^1$. Given this strategy, the short-lived players' posterior probability of the commitment type at history $h_t$ is
\begin{align}
\label{eq:mut}
\mu_t^{\gs^1}(h_t)
:=
\frac{\mu_0 F_t(h_t)}{\mu_0F_t(h_t)+(1-\mu_0) L_t^{\gs^1}(h_t)}.
\end{align}
Their period-$t$ predictive probability of signal $y$ is 
\begin{align*}
q_t^{\gs^1}(y\mid h_t):=\mu_t^{\gs^1}(h_t)F(y\mid h_t)+(1-\mu_t^{\gs^1}(h_t))\rho_{\gs_t^1(h_t,\xi^0)}(y).
\end{align*}
Because each period player 2 optimizes in response to her predictive signal distribution in that period, and because this distribution may lie outside the set of feasible distributions $\{\rho_\ga:\ga\in\Delta(A)\}$ when $F$ is misspecified, it is convenient to define player 2's best response directly as a function of her predictive signal distribution.\footnote{When monitoring also depends on player 2's action, an entropy-confirming notion is required, as is standard; see \cref{sec:slp-dependent-monitoring}.} Specifically, for any predictive signal distribution $q\in\Delta(Y)$, define player 2's best-response set by $\BR^2(q):=\argmax_{\gb\in\Delta(B)}\bE_{y\sim q}\!\left[\tilde v(\gb,y)\right]$. 

I next define equilibrium. Under strategy profile $\gs$, let
\begin{align*}
U(\gs;\gd,F)
:=
\bE_\gs\!\left[
(1-\gd)\sum_{t=0}^\infty\gd^tu(a_t,b_t)
\mathrel{\Big|}\xi^0
\right]
\end{align*}
be the normal type's (discounted-)average payoff when the discount factor is $\gd$ and the subjective process is $F$.

\begin{def*}\label{def:eqm}
A strategy profile $\gs = (\gs^1, \gs^2)$ is a \emph{(subjective Nash) equilibrium} if (i) $\gs^1(\cdot,\xi^0)$ maximizes the normal type's average payoff against $\gs^2$ and (ii) $\gs_t^2(h_t)\in\BR^2(q_t^{\gs^1}(\cdot\mid h_t))$ for every $t$ and every $h_t\in Y^t$.
\end{def*}

My solution concept is tailored to the robustness question: it retains the canonical framework's Bayesian updating and equilibrium requirements while varying only the subjective process. The resulting payoff comparisons therefore isolate the effect of this misspecification. It differs from Berk--Nash equilibrium \citep{esponda2016berk}, in which equilibrium beliefs are supported on specifications within a given subjective model class that minimize divergence from observations generated by equilibrium behavior. 

Let $\Gs(\gd,F)$ denote the set of equilibria. An equilibrium exists by a straightforward adaptation of the truncation argument in \citet{fudenberg1983subgame}.\footnote{Their setting has a fixed finite set of players who evaluate outcomes under a common stochastic process. Here there is one long-lived player and many short-lived players, and the two sides may evaluate histories under different processes. The truncation argument nevertheless applies: every finite truncation has finitely many active players and is a finite strategic-form game, average payoffs are continuous at infinity, and the relevant best-response conditions are closed.} Define the lower bound and the upper bound of the normal type's equilibrium payoff as follows:
\begin{align*}
\underline W(\gd;F)
&:=\inf_{\gs\in\Gs(\gd,F)}U(\gs;\gd,F),
&
\overline W(\gd;F)
&:=\sup_{\gs\in\Gs(\gd,F)}U(\gs;\gd,F).
\end{align*}

While my model retains the canonical framework's objective primitives and therefore treats $F$ as a misspecification whenever $F\neq\widehat P$, it also admits a correctly specified interpretation with the same equilibria and normal-type payoffs, which my results focus on. More precisely, if the commitment type's signals were assumed to objectively follow the process $F$ rather than $\widehat P$, the short-lived players' posterior beliefs and predictions, and the normal type's equilibrium payoffs and incentives would be unchanged.

As motivated in the Introduction, I compare subjective processes by what can be assessed from observables. Because players observe only a finite history at any time, this comparison concerns finite-horizon distributions:

\begin{def*}\label{def:vanishempmis}
Misspecification \emph{vanishes empirically} along a sequence $(F^n)_{n=1}^{\infty}$ in $\Delta(H)$ if, for every fixed horizon $T=1,2,\dots$, $F_T^n$ converges to $\widehat P_T$ in total variation.\footnote{For probability distributions $Q,Q'$ on a finite set $X$, the total-variation norm is $\|Q-Q'\|_{\mathrm{TV}}:=\max_{E\subseteq X}|Q(E)-Q'(E)|=\tfrac12\sum_{x\in X}|Q(x)-Q'(x)|$.} The \emph{finite-history topology} is the topology on $\Delta(H)$ in which $F^n\ra F$ if and only if $F_T^n$ converges to $F_T$ in total variation for every fixed $T$.
\end{def*}

Because $Y^T$ is finite for each $T$, total-variation, weak, and pointwise convergence coincide in \cref{def:vanishempmis}. Both misspecification examples in \cref{sec:productchoice} vanish empirically. The finite-history topology is standard: since $Y$ is finite, it is the weak topology of probability measures on the compact metrizable space $H$, and so $\Delta(H)$ is compact and metrizable. Imposing any metric that induces this topology would change none of my results, which depend on the topology alone. 

My comparison of subjective processes follows the approach of \citet{weinstein2007structure}, who compare belief hierarchies by what can be assessed from observables. Indeed, identifying each $F$ with its sequence of finite-horizon marginals $(F_T)_{T=1}^{\infty}$, the finite-history topology is the product topology, in which convergence is coordinatewise convergence of these marginals; it is the analogue, for signal processes, of the product topology on belief hierarchies in \citet{weinstein2007structure}, in which convergence is order-by-order convergence of beliefs.\footnote{Distributions can converge at every fixed horizon while remaining far apart at horizons that grow along the sequence. A similar distinction arises in the literature on higher-order beliefs; see \egg \citet{dekel2006topologies} and \citet{chen2010uniform}.}

\cref{def:vanishempmis} has a direct statistical interpretation. Let $\varphi:Y^T\ra[0,1]$ denote a test that rejects the hypothesis of signals following $\widehat P$ with probability $\varphi(h_T)$ given history $h_T$, so that its rejection probability under a process $F$ is $\bE_F[\varphi]$. For each $F^n$, the largest difference in rejection probabilities across such tests equals the total-variation distance: $\sup_{\varphi}|\bE_{F^n}[\varphi]-\bE_{\widehat P}[\varphi]|=\|F_T^n-\widehat P_T\|_{\mathrm{TV}}$. Thus, misspecification vanishes empirically exactly when rejection probabilities converge to those under the true process, uniformly over all tests at each fixed horizon.

\section{Benchmark}
\label{sec:canonical}

This section revisits the canonical benchmark. Let $\widehat\rho:=\rho_{\widehat\ga}$ denote the signal distribution generated by the commitment action $\widehat \ga$. Define the commitment payoff---player 1's worst payoff from playing $\widehat\ga$ against some player 2's best response to $\widehat \rho$---by
\begin{align}
\label{eq:R}
R
:=
\min_{\gb\in\BR^2(\widehat\rho)}u(\widehat\ga,\gb).
\end{align}

\begin{bthm*}[\citealp{fudenberg1992}; \citealp{gossner2011simple}] \label{thm:fl}
It holds that
\begin{align}
\label{eq:FLbound}
\liminf_{\gd\ra1}\underline W(\gd;\widehat P)\geq R.
\end{align}
\end{bthm*}

Therefore, under correct specification, every equilibrium gives a sufficiently patient normal type at least approximately the commitment payoff $R$. The intuition is by now standard. In the cleanest case where the commitment action is pure and statistically identifiable---that is, no other (potentially mixed) action induces the same signal distribution---\citet[Section 2.4]{mailath2015reputations} offer the following intuition. Fix any equilibrium, and suppose the normal type deviates to persistent play of the commitment action. If the short-lived players assign low probability to that action being played, then the signals generated by the deviation tend to increase their posterior belief in the commitment type, and hence the posterior probability that the commitment action is being played. Thus, there cannot be too many histories at which they assign low probability to the commitment action being played. A patient normal type treats such histories as negligible and therefore, by deviating, obtains approximately the commitment payoff $R$. The proof is standard and omitted. Proofs of other formal statements appear in \cref{sec:proofs} unless noted otherwise.

To state the reputation effect, let $\Gs^{CI}(\gd)$ be the set of equilibria under complete information when player 1 is commonly known to be normal, and let $U^{CI}(\gs;\gd)$ be his payoff in equilibrium $\gs\in\Gs^{CI}(\gd)$. Denote the patient upper limit of the normal type's highest complete-information equilibrium payoff by
\begin{align}
\label{eq:barWCIa}
\overline W^{CI}
:=
\limsup_{\gd\ra1}
\sup_{\gs\in\Gs^{CI}(\gd)}U^{CI}(\gs;\gd).
\end{align}

Provided that $R > \overline W^{CI}$ so that commitment is valuable to the normal type, as in the product-choice game in \cref{sec:productchoice}, the reputation effect arises---his payoff guarantee strictly exceeds his highest complete-information equilibrium payoff:

\begin{cor*}[Reputation effect]\label{cor:re}
If $R>\overline W^{CI}$, then
\begin{align*}
\liminf_{\gd\ra1}\underline W(\gd;\widehat P)>\overline W^{CI}.
\end{align*}
\end{cor*}

I omit a formal proof of this corollary, since it follows directly from \cref{thm:fl}.

\section{Collapse}
\label{sec:collapse}

In this section I show that the reputation effect need not survive vanishing empirical misspecification. Denote the set of subjective processes for which the patient upper limit of the normal type's highest equilibrium payoff is no greater than the complete-information benchmark---that is, the reputation effect collapses---by
\begin{align*}
\mathcal F^{\mathrm{collapse}} :=
\left\{ F\in\Delta(H): \limsup_{\gd\ra1}\overline W(\gd;F) \le \overline W^{CI} \right\}.
\end{align*}

\begin{thm*}[Collapse]
\label{thm:fd-collapse}
The set $\mathcal F^{\mathrm{collapse}}$ is dense in $\Delta(H)$ in the finite-history topology. In particular, there exists a sequence $(F^n)_{n=1}^{\infty}$ of full-support processes for which misspecification vanishes empirically such that for every $n=1,2,\dots$,
\begin{align}
\limsup_{\gd\ra1}\overline W(\gd;F^n)
\leq
\overline W^{CI}.
\label{eq:collapse}
\end{align}
Moreover, the sequence $(F^n)_{n=1}^{\infty}$ can be chosen so that $F_T^n=\widehat P_T$ for every $n$ and every horizon $T=1,\dots, n$.
\end{thm*}

Thus, every neighborhood of every subjective process, correctly specified or not, in the finite-history topology contains a subjective process under which the reputation effect collapses. The reason is that the finite-history topology imposes little restriction on arbitrarily long horizons, which are payoff-relevant for a patient normal type, enabling misspecifications that vanish empirically yet have a large effect on his payoff.

\cref{thm:fd-collapse} is stronger than merely identifying a failure of the reputation effect; the latter would mean 
\begin{align}
\liminf_{n\ra\infty}\liminf_{\gd\ra1}\underline W(\gd;F^n) \le \overline W^{CI}.
\label{eq:empirical-robustness}
\end{align}
The strengthening is two-fold. First, \eqref{eq:collapse} upper-bounds the patient upper limit of the normal type's highest equilibrium payoff, not only the patient lower limit of his worst equilibrium payoff. Second, \eqref{eq:collapse} occurs under every subjective process in the sequence, not only in the limit. 

The collapse is a patient-limit result. At a fixed, sufficiently high discount factor, sufficiently slight misspecification may still permit reputation gains, because the histories exposing it may receive little discounted weight in the normal type's payoff. For each fixed process inducing collapse, however, any positive reputation gain becomes arbitrarily small as the discount factor approaches one.\footnote{This is reminiscent of the bad-reputation result of \citet{ely2003bad}, where patience likewise washes out transitory reputation gains. The mechanism is different: bad reputation arises from short-lived players' willingness to interact with the long-lived player, whereas the collapse here is driven by their subjective model of commitment.} Formally:

\begin{prop*}
\label{prop:patience-hurts}
Suppose $R>\overline W^{CI}$, and let $(F^n)_{n=1}^{\infty}$ be a sequence in which $F^n \in \mathcal F^{\mathrm{collapse}}$ for each $n$, and along which misspecification vanishes empirically. There exist $\gd_0\in(0,1)$ and $n_0\geq1$ such that, for every $n\ge n_0$,
\begin{align}
\underline W(\gd_0;F^n)
>
\overline W^{CI}
\geq
\limsup_{\gd\ra1}\overline W(\gd;F^n).
\label{eq:patience-hurts}
\end{align}
\end{prop*}

In the rest of this section, I sketch the proof of \cref{thm:fd-collapse}. Because the normal type can adjust his behavior to the short-lived players' misspecification, the argument must run through equilibrium rather than through a comparison of signal processes alone. I proceed in two steps. 

\paragraph{Step 1.} I derive a sufficient condition for the collapse.

\begin{prop*}
\label{prop:genfalsification}
Let $F\in\Delta(H)$. Suppose that, for every $K>0$, there are deterministic dates $m_{\gd,K}\in\{1,2,\ldots\}$ such that, as $\gd\ra1$,\footnote{The log-likelihood ratio in \eqref{eq:early-separation-likelihood} below is defined as $\infty$ when its denominator is zero.}
\begin{align}
(1-\gd)m_{\gd,K}&\ra 0,
\label{eq:early-separation-date}\\
\inf_{\gs^1}L^{\gs^1}\!\left(
\log\frac{L_{m_{\gd,K}}^{\gs^1}(h_{m_{\gd,K}})}
{F_{m_{\gd,K}}(h_{m_{\gd,K}})}
\geq\frac{K}{1-\gd} 
\right)&\ra 1.
\label{eq:early-separation-likelihood}
\end{align}
Then $F\in\mathcal F^{\mathrm{collapse}}$.
\end{prop*}

Condition \eqref{eq:early-separation-date} identifies dates early enough that preceding payoffs are negligible as patience grows. Condition \eqref{eq:early-separation-likelihood} ensures that, with high probability, the posterior odds of the normal type against the commitment type remain high for long enough that later play has a negligible effect on the normal type's payoff. While these odds remain high, the short-lived players approximately best respond to the normal type. His equilibrium payoff in the patient upper limit is thus capped at $\overline W^{CI}$. In general, approximate optimization in dynamic games can have large payoff effects (see, \egg \citealp{radner1980collusive}); the bound nonetheless holds here because the approximation is confined to the short-lived players' behavior.

The conditions in \cref{prop:genfalsification} restrict only the likelihoods of long histories and are therefore compatible with misspecification that vanishes empirically. The gradual-improvement misspecification in \cref{sec:productchoice} illustrates this possibility: as the firm matures, consumers regard low-quality signals as increasingly unlikely under commitment and thus as strong evidence that the firm is a normal type. In \cref{sec:approxfalcal}, I verify that for each $n$, the misspecification in that example satisfies the conditions of the proposition, leading to the collapse of the reputation effect.

\paragraph{Step 2.} This step constructs nearby processes and completes the proof of \cref{thm:fd-collapse}. Fix any subjective process $G$. For each $n$, I construct a process $F^n$ that agrees with $G$ through horizon $n$ but differs thereafter. Beyond date $n$, the probability conditional on any history that $F^n$ assigns to some given signal (which occurs infinitely often almost surely under every normal-type strategy, since monitoring has full support) declines rapidly enough to satisfy the conditions of \cref{prop:genfalsification}, while remaining positive. This shows that $\mathcal F^{\mathrm{collapse}}$ is dense. 

By construction, whenever $G$ has full support, so does each $F^n$. Because the construction applies in particular to $G = \widehat P$, which has full support, it gives a sequence of full-support processes $(F^n)_{n=1}^\infty$, where each $F^n$ agrees with $\widehat P$ through horizon $n$, along which misspecification vanishes empirically but the reputation effect collapses under $F^n$ for every $n$.

\begin{rem*}[The scope of \cref{prop:genfalsification}]\normalfont
\Cref{prop:genfalsification} includes the simpler case of perfect falsification, where some histories rule out commitment perfectly. For example, consumers in the product-choice game might believe that the commitment type follows the true signal process except that, after $n-1$ consecutive low-quality signals, the next signal is certainly high-quality. This process rules out $n$ consecutive low-quality signals and agrees with the true process over every horizon shorter than $n$, so misspecification vanishes empirically as $n\ra\infty$. \Cref{sec:exactimplies} defines perfect falsification and shows that it implies the hypotheses of \cref{prop:genfalsification}.

\cref{prop:genfalsification} gives a sufficient but not necessary condition for collapse. In \cref{sec:vanishingaverage}, I give an alternative sufficient condition under which the normal type's payoff advantage from the short-lived players' belief in commitment vanishes in discounted average as patience grows, although their posterior probability of commitment need not become small. This condition can also hold when subjective signal probabilities conditional on any history are uniformly bounded away from zero, a case excluded by the conditions of \cref{prop:genfalsification}. A necessary-and-sufficient characterization would need to capture the interaction between likelihoods and equilibrium behavior, a task beyond the scope of this paper.
\end{rem*}

\section{Fragility to subjective inference}
\label{sec:norobust}

This section shows that the collapse of the reputation effect is no more robustly predictable than its survival. Let
\begin{align*}
\mathcal F^{\mathrm{survival}}
:=
\left\{
F\in\Delta(H):
\liminf_{\gd\ra1}\underline W(\gd;F)
>\overline W^{CI}
\right\}
\end{align*}
be the set of subjective processes for which the reputation effect survives. Also let 
\begin{align*}
\mathcal F^{\mathrm{canonical}}
:=
\left\{
F\in\Delta(H):
\liminf_{\gd\ra1}\underline W(\gd;F)
\ge R
\right\}
\end{align*}
be the set of subjective processes for which the patient normal type secures the canonical reputation guarantee. If $R>\overline W^{CI}$ so that the reputation effect is present in the canonical framework, then $\mathcal F^{\mathrm{canonical}} \subseteq \mathcal F^{\mathrm{survival}}$.

\begin{thm*}[Fragility to subjective inference]
\label{thm:nowherecontinuity}
Suppose that $R>\overline W^{CI}$. The set $\mathcal F^{\mathrm{canonical}}$, and therefore $\mathcal F^{\mathrm{survival}}$, is dense in $\Delta(H)$ in the finite-history topology. Together with the density of $\mathcal F^{\mathrm{collapse}}$ established in \cref{thm:fd-collapse}, this implies that the maps
\begin{align}\label{eq:patientmaps}
F\mapsto\limsup_{\gd\ra1}\overline W(\gd;F)
\qquad\text{and}\qquad
F\mapsto\liminf_{\gd\ra1}\underline W(\gd;F)
\end{align}
are nowhere continuous.
\end{thm*}

Therefore, in the finite-history topology, every neighborhood of every subjective process contains both a process under which the reputation effect collapses and one under which it survives. A sufficiently patient long-lived player's reputation gains are fragile to the short-lived players' subjective inference.

The proof uses relative entropy. For distributions $Q,Q'$ on the same finite set, let $\KL{Q}{Q'}:=\bE_Q[\log(Q/Q')]$ denote the relative entropy of $Q$ relative to $Q'$. This is the expected log-likelihood loss from evaluating observations generated under $Q$ using $Q'$. It takes values in the extended reals, with $0\log(0/x)=0$ and $x\log(x/0)=\infty$ for $x>0$. For any $F\in\Delta(H)$, define the \emph{entropy rate} of $\widehat P$ relative to $F$ by
\begin{align}
\overline D(\widehat P\parallel F)
:=
\limsup_{T\to\infty}
\frac{1}{T}
\KL{\widehat P_T}{F_T}.
\label{eq:entropyrate}
\end{align}

I next sketch the proof, proceeding in three steps.

\paragraph{Step 1.} I show that a grain-of-truth condition, namely that the subjective process assigns a positive weight to the true process, ensures a vanishing entropy rate:

\begin{prop*}
\label{prop:standard-misspecifications}
Let $F \in \Delta(H)$. Suppose that $F=\lambda \widehat P+(1-\lambda)G$ for some $G \in\Delta(H)$ and $\lambda\in(0,1]$. Then $\overline D(\widehat P\parallel F) = 0$.
\end{prop*}

\cref{prop:standard-misspecifications} follows because the grain-of-truth condition implies an upper bound on the finite-horizon relative entropy $\KL{\widehat P_T}{F_T}$ uniformly over all horizons $T$, which then implies a zero entropy rate $\overline D(\widehat P \parallel F)$. 

The grain-of-truth condition here differs from that of \citet{sorin1999merging}, which concerns the short-lived players' unconditional signal process---the mixture of the commitment-type and normal-type processes. In the canonical framework, his condition holds by assumption: when the normal type persistently plays the commitment action, the true signal process is $\widehat P$, and the unconditional process $\mu_0\widehat P+(1-\mu_0)L^{\gs^1}$ puts weight $\mu_0$ on it. Under misspecification, the unconditional process is $\mu_0F+(1-\mu_0)L^{\gs^1}$ and need not put any weight on $\widehat P$; the condition $F=\lambda \widehat P+(1-\lambda)G$ for some $G \in\Delta(H)$ here restores the grain, with weight $\mu_0\lambda$.

\paragraph{Step 2.} I next show that the canonical reputation guarantee is preserved along any sequence of subjective processes whose entropy-rate loss vanishes:\footnote{An analogous robustness result holds for an undiscounted finite-horizon counterpart by taking the horizon to be arbitrarily long; see \cref{sec:finitehorizon}.}

\begin{prop*}
\label{prop:entropy-rate-robustness}
Let $(F^n)_{n=1}^{\infty}$ be any sequence in $\Delta(H)$. If $\overline D(\widehat P\parallel F^n)\ra0$ as $n \ra \infty$, then
\begin{align}
\liminf_{n\to\infty}
\liminf_{\gd\ra1}
\underline W(\gd;F^n)
\geq R.
\label{eq:entropyrobust}
\end{align}
\end{prop*}

The order of limits in \eqref{eq:entropyrobust} captures the robustness question of interest, asking whether the canonical reputation guarantee survives misspecification with vanishing entropy rate.\footnote{A joint-limit version of \eqref{eq:entropyrobust} can be established. See \cref{sec:jointlimit}.} Vanishing entropy-rate loss controls average likelihood loss of evaluating signals generated by $\widehat P$ using $F^n$ over long horizons, which finite-history proximity alone does not control. This control preserves the canonical reputation guarantee. 

\cref{prop:entropy-rate-robustness} speaks to the distinction between the two misspecification examples in \cref{sec:productchoice}. In \cref{sec:err}, I apply \cref{prop:entropy-rate-robustness} to show that, unlike the first misspecification example, the second satisfies this vanishing entropy-rate condition and therefore preserves the canonical reputation guarantee.

To prove this proposition, I consider the canonical deviation by the normal type to persistent play of the commitment action and extend \citeauthor{gossner2011simple}'s (\citeyear{gossner2011simple}) entropy-to-payoff argument to accommodate the mismatch due to misspecification: the short-lived players form their equilibrium predictions using the subjective process $F$, whereas the normal type's deviation generates histories according to $\widehat P$. Define the average entropy loss under the normal type's deviation to persistent play of $\widehat \ga$ by
\begin{align}
D_\gd(\widehat P\parallel F)
:=
(1-\gd)\sum_{t=0}^\infty\gd^t
\bE_{\widehat P}\!\left[
\KL{\widehat\rho}{F(\cdot\mid h_t)}
\right].
\label{eq:pdentropyrate}
\end{align}
In every equilibrium $\gs$, Bayes' rule and the chain rule for relative entropy imply
\begin{align}
\bE_{\widehat P}\!\left[
(1-\gd)\sum_{t=0}^\infty\gd^t
\KL{\widehat\rho}{q_t^{\gs^1}(\cdot\mid h_t)}
\right]
\leq
-(1-\gd)\log\mu_0
+
D_\gd(\widehat P\parallel F).
\label{eq:discounted-entropy-budget}
\end{align}
Persistent play of $\widehat\ga$ gives the normal type nearly $R$ or more whenever the short-lived players' predictive signal distribution is close to $\widehat\rho$. The entropy budget---the right side of \eqref{eq:discounted-entropy-budget}---limits the expected discounted frequency of predictions that give him substantially less than $R$. A vanishing budget then restores the canonical guarantee. In the budget, the prior cost $-(1-\gd)\log\mu_0$ vanishes with patience. The misspecification cost $D_\gd(\widehat P\parallel F)$ is a discounted average of expected one-period likelihood losses. As patience grows, its patient upper limit is bounded by the long-run average loss, so $\limsup_{\gd\ra1}D_\gd(\widehat P\parallel F)\leq\overline D(\widehat P\parallel F)$. Thus, given a sequence of subjective processes $(F^n)_{n=1}^{\infty}$ along which the entropy rate $\overline D(\widehat P\parallel F^n)$ vanishes, taking the patient limit for each $F^n$ and then letting $n\ra\infty$ collapses the budget, thereby restoring the canonical reputation guarantee for the patient normal type.

\paragraph{Step 3.} This step completes the proof of \cref{thm:nowherecontinuity}. Given any signal process $G \in \Delta(H)$ and $\gl \in (0,1)$, consider the mixture $F^\lambda:=\lambda\widehat P+(1-\lambda)G$. It has zero entropy-rate loss by \cref{prop:standard-misspecifications}. Applying \cref{prop:entropy-rate-robustness} to the constant sequence $(F^n)_{n=1}^\infty$, in which $F^n=F^\lambda$ for each $n$, therefore gives $\liminf_{\gd\ra1}\underline W(\gd;F^\lambda)\ge R$. Since $F^\lambda\ra G$ as $\lambda\searrow0$ and $G$ is arbitrary, $\mathcal F^{\mathrm{canonical}}$ is dense. Hence, $\mathcal F^{\mathrm{survival}}$ is also dense when $R>\overline W^{CI}$. Together with the density of $\mathcal F^{\mathrm{collapse}}$ from \cref{thm:fd-collapse}, this shows that every nonempty open set in the finite-history topology contains a subjective process at which each map in \eqref{eq:patientmaps} is at least $R$ and another at which each is at most $\overline W^{CI}$. Within every neighborhood, each map therefore varies by at least $R-\overline W^{CI}>0$, so both maps are nowhere continuous.

I close this section with two remarks.

\begin{rem*}[Scope of \cref{prop:standard-misspecifications}]\normalfont\label{rem:diffusebayesianuncertainty}
The grain-of-truth condition in \cref{prop:standard-misspecifications} is sufficient for zero entropy-rate loss but is not necessary. To see this, consider a prior distribution $\pi$ over full-support marginals $f\in\Delta(Y)$ that assigns positive probability to every neighborhood of the true per-period commitment-type signal distribution $\widehat\rho$ but satisfies $\pi(\{\widehat\rho\})=0$. Suppose the short-lived players believe that signals are i.i.d.\ conditional on $f$; their subjective process $F$ then has marginals $F_T=\int f^{\otimes T}\pi(\diff f)$ for every finite horizon $T$, where $f^{\otimes T}$ denotes the distribution of $T$ independent signals with common marginal $f$. The process $F$ violates the grain-of-truth condition: since $f$ is drawn once from $\pi$, by the strong law of large numbers, the event $E$ that the observed frequency of each signal $y\in Y$ converges to $\widehat\rho(y)$ satisfies $\widehat P(E)=1$ and $F(E)=\pi(\{\widehat\rho\})=0$, and therefore $F$ cannot contain $\widehat P$ with positive mixture weight. Nonetheless, in \cref{sec:appendix-diffusebayesianuncertainty}, I show that the entropy rate $\overline D(\widehat P \parallel F)$ is equal to zero, and so \cref{prop:entropy-rate-robustness} implies that the canonical reputation guarantee survives: $\liminf_{\gd\ra1}\underline W(\gd;F)\geq R$.
\end{rem*}

\begin{rem*}[Noncanonical deviation]\normalfont
In proving \cref{prop:entropy-rate-robustness}, considering the canonical deviation to persistent play of the commitment action suffices. Under misspecification, however, another deviation may better match the subjective process. \Cref{sec:genrepbuild} derives an equilibrium payoff lower bound by optimizing over deviations, accounting for both their payoffs and their statistical mismatch with that process, and gives an example showing that other deviations can improve the payoff bound under misspecification. The bound recovers the canonical guarantee under correct specification.
\end{rem*}

\section{Fragility to discounting}
\label{sec:generic-instability}

In this section, I show that when $R > \overline W^{CI}$ so that the reputation effect is present in the canonical framework, for a generic subjective process, reputation gains are sensitive to the discount factor, however close it is to one, in sharp contrast with the canonical prediction that holds for all sufficiently high discount factors. Recall that a set is meagre if it is a countable union of nowhere dense sets, and comeagre if its complement is meagre. A property is generic if the set of processes satisfying it is comeagre, and a process is generic if it belongs to this set. This is the standard topological notion of genericity.\footnote{I do not adopt a probabilistic notion. The reason is that a claim about the frequency of collapse or survival would require assuming a reference probability measure, and the conclusion is sensitive to its choice. In \cref{sec:proof-reference-measures}, I prove that, whenever $R>\overline W^{CI}$, for every $p\in[0,1]$ there exists a nonatomic Borel probability measure $\Pi_p$ on $\Delta(H)$, with full support in the finite-history topology, such that $\Pi_p(\mathcal F^{\mathrm{collapse}})=p$; full support of a reference measure means that every nonempty open set of subjective processes has positive probability. Thus, the probability of collapse is unrestricted.}

\begin{thm*}[Fragility to discounting]
\label{thm:generic-instability}
The set
\begin{align}
\mathcal G
:=
\left\{
F\in\Delta(H):
\liminf_{\gd\ra1}\overline W(\gd;F)
\leq\overline W^{CI}
\text{ and }
\limsup_{\gd\ra1}\underline W(\gd;F)
\geq R
\right\}
\label{eq:generic-instability}
\end{align}
is comeagre (in the finite-history topology) so that processes in $\mathcal G$ are generic. Therefore, if $R>\overline W^{CI}$, then $\mathcal F^{\mathrm{collapse}}$ and $\mathcal F^{\mathrm{survival}}$ are meagre.
\end{thm*}

\cref{thm:generic-instability} shows that for a generic subjective process, all equilibria secure the canonical reputation guarantee in the limit along some sequences of discount factors tending to one, while even the best equilibrium payoff is at most the complete-information benchmark in the limit along others. Therefore, when $R>\overline W^{CI}$, substantial reputation gains at a high discount factor for the normal type need not persist at discount factors closer to one, in sharp contrast with the canonical prediction that the reputation effect emerges for all sufficiently high discount factors. Consequently, while the density of both collapse and survival of the reputation effect leaves open whether either prediction is typical, \cref{thm:generic-instability} shows that neither is.

The key observation in proving the theorem is that strict payoff bounds at a fixed discount factor survive sufficiently small changes in the subjective process. Within any neighborhood of subjective processes, the density of $\mathcal F^{\mathrm{canonical}}$ gives a process at which the normal type's payoff in every equilibrium exceeds a target just below $R$ at some arbitrarily high discount factor. This bound continues to hold throughout a smaller neighborhood containing this subjective process. Within that smaller neighborhood, the density of $\mathcal F^{\mathrm{collapse}}$ gives another process at which even the highest equilibrium payoff is strictly below a target just above $\overline W^{CI}$ at a still higher discount factor, while preserving the first bound at the original discount factor. Both strict bounds then hold throughout a sufficiently small neighborhood of this new process. Consequently, for any prescribed discount-factor threshold below one, and however close the two payoff targets are to their respective benchmarks $R$ and $\overline W^{CI}$, the set of subjective processes for which both comparisons hold, at different discount factors above that threshold, is open and dense. Applying the Baire category theorem \citep[Chapter 8, Theorem 48.2]{munkres2000topology} to countably many thresholds approaching one and targets approaching their respective benchmarks shows that, for a generic fixed subjective process, both payoff comparisons recur at discount factors arbitrarily close to one, with targets arbitrarily close to $R$ and $\overline W^{CI}$.

\section{Entropy-rate robustness}
\label{sec:err}

My main results imply that when the short-lived players may misspecify the commitment-type signal process, the canonical reputation guarantee rests on restrictions about how they interpret long histories---those that no fixed amount of finite-history evidence about the commitment-type signal process can justify. \cref{prop:entropy-rate-robustness} identifies one such restriction: vanishing entropy-rate loss controls average likelihood loss over long horizons and preserves the guarantee.

In this final section, I develop two further implications of this entropy-rate criterion. First, slight misspecification in the form of i.i.d.\ marginals or fixed-order Markov kernels is innocuous for the canonical reputation guarantee:

\begin{prop*}
\label{prop:standard-misspecifications2}
Let $(F^n)_{n=1}^{\infty}$ be a sequence in $\Delta(H)$. Suppose that for some fixed $k =0,1,\dots$, each $F^n$ is a time-homogeneous Markov process of order $k$, with a full-support initial distribution and transition kernel $\pi^n$, and $\sup_{h_k\in Y^k}D(\widehat\rho\parallel\pi^n(\cdot\mid h_k))\ra0$ as $n \ra \infty$. When $k=0$, $F^n$ is i.i.d.\ with common signal distribution $\pi^n$. Then $\overline D(\widehat P\parallel F^n)\ra0$ and hence \eqref{eq:entropyrobust} holds.
\end{prop*}

Intuitively, for each fixed $n$, the finite loss from the first $k$ signals does not affect the entropy rate. In each subsequent period, the prediction depends on the preceding $k$ signals, but its relative-entropy loss is at most $\sup_{h_k\in Y^k}D(\widehat\rho\parallel\pi^n(\cdot\mid h_k))$. This also bounds the entropy-rate loss and tends to zero as $n\ra\infty$, proving \cref{prop:standard-misspecifications2}. The second misspecification in \cref{sec:productchoice} illustrates this proposition.

Second, reputation effects do not require the normal type to be able to generate the subjective process. As the proof sketch of \cref{prop:entropy-rate-robustness} explains, what preserves the canonical guarantee is that, under the deviation to persistent play of the commitment action, the short-lived players' predictive signal distributions stay close to $\widehat\rho$ in most periods. This can hold even when their belief in the commitment type is negligible, so long as they expect the normal type to play a (potentially mixed) action that generates a signal distribution close to $\widehat\rho$. To illustrate this observation, consider a class of misspecifications that generalizes the second misspecification example in \cref{sec:productchoice}. Call a subjective process $F$ \emph{i.i.d.\ commitment-separating} if there exists a full-support $f\in\Delta(Y)$ such that $F_T(h_T)=\prod_{t=0}^{T-1}f(y_t)$ for every $T$ and $h_T=(y_0,\ldots,y_{T-1})$, and $\inf_{\ga\in\Delta(A)}
\KL{\rho_\ga}{f}
>0$. For any equilibrium $\gs$, let $\ell_t^\gs(h_t)
:=
\max_{b\in B}
v(\gs_t^1(h_t,\xi^0),b)
-
v(\gs_t^1(h_t,\xi^0),\gs_t^2(h_t))$ be the short-lived player's loss relative to a best response to the normal type's equilibrium action.

\begin{prop*}
\label{prop:vanish}
If $F$ is i.i.d.\ commitment-separating, then
\begin{align}
\limsup_{\gd\ra1}
\sup_{\gs\in\Gs(\gd,F)}
\bE_\gs
\!\left[
(1-\gd)
\sum_{t=0}^{\infty}
\gd^t
\mu_t^{\gs^1}(h_t)
\middle|
\xi^0
\right]
&=0,
\label{eq:vanishingrep}\\
\limsup_{\gd\ra1}
\sup_{\gs\in\Gs(\gd,F)}
\bE_\gs
\!\left[
(1-\gd)
\sum_{t=0}^{\infty}
\gd^t
\ell_t^\gs(h_t)
\middle|
\xi^0
\right]
&=0.
\label{eq:avglossboundmain}
\end{align}
\end{prop*}

\cref{prop:vanish} shows that, in any equilibrium, commitment beliefs and the short-lived players' losses from failing to best respond to the normal type vanish in discounted average as patience grows.\footnote{This differs from \citet{cripps2004imperfect}, who fix the discount factor and examine the posterior belief in the commitment type as time tends to infinity.} Nevertheless, \cref{prop:entropy-rate-robustness} shows that the canonical guarantee can survive along sequences of such processes with vanishing entropy-rate loss. Reputation gains can therefore survive even when belief in the commitment type vanishes under equilibrium play.

\appendix
\appendixtitleon

\pagebreak

\begin{appendices}

\section{Omitted details}
\label{sec:om}



\subsection{Gradual-improvement example and \cref{prop:genfalsification}}
\label{sec:approxfalcal}

Here I verify that the gradual-improvement misspecification in \cref{sec:productchoice} satisfies the hypotheses in \cref{prop:genfalsification}. Fix $n$ and $K>0$. Set $r_\gd:=\lceil n\sqrt{4K/(1-\gd)}\rceil$, $m_\gd:=2r_\gd+1$, and $\underline\rho:=\min\{p',1-p\}>0$. Let $A_\gd$ be the event that at least one low-quality signal $y_l$ occurs at dates $r_\gd,\ldots,2r_\gd$. Under every normal-type strategy, the conditional probability of a low-quality signal is at least $1-p$ after every history. The probability of avoiding it throughout these $r_\gd+1$ periods is therefore at most $p^{r_\gd+1}$. Thus $\inf_{\gs^1}L^{\gs^1}(A_\gd)\geq1-p^{r_\gd+1}\ra1$ as $\gd \ra 1$.

On $A_\gd$, a low-quality signal $y_l$ occurs at some time $t\in\{r_\gd,\ldots,2r_\gd\}$. Under $F^n$, the probability of this signal at that date is $(1-p)\expo^{-(t/n)^2}\leq\expo^{-(r_\gd/n)^2}$. The probability of the entire history is the product of its signal probabilities, bounded above by one. Therefore $F_{m_\gd}^n(h_{m_\gd})\leq\expo^{-(r_\gd/n)^2}$. Under any normal-type strategy, by contrast, the conditional probability of each observed signal is at least $\underline\rho$. The history contains $m_\gd$ signals, so $L_{m_\gd}^{\gs^1}(h_{m_\gd})\geq\underline\rho^{m_\gd}$. Combining the two bounds gives $\log[L_{m_\gd}^{\gs^1}(h_{m_\gd})/F_{m_\gd}^n(h_{m_\gd})]\geq(r_\gd/n)^2-m_\gd|\log\underline\rho|$ on $A_\gd$.

By the choice of $r_\gd$, $(r_\gd/n)^2\geq4K/(1-\gd)$. Also, $(1-\gd)m_\gd\ra0$ as $\gd\ra1$, with $n$ and $K$ fixed. Therefore, $(1-\gd)\log[L_{m_\gd}^{\gs^1}(h_{m_\gd})/F_{m_\gd}^n(h_{m_\gd})]\geq4K-(1-\gd)m_\gd|\log\underline\rho|\geq K$ for all $\gd$ sufficiently close to one. Together with the uniform probability bound for $A_\gd$ and $(1-\gd)m_\gd\ra0$, the hypotheses of \cref{prop:genfalsification} hold.

\subsection{Perfect falsification and likelihood separation}
\label{sec:exactimplies}

Say that a subjective process $F$ is \emph{perfectly falsifiable} if there exists a stopping time $\theta:H\to\{0,1,\ldots\}\cup\{\infty\}$, measured with respect to the signal histories, such that $F(\theta<\infty)=0$ and $\sup_{\gs^1}\bE_{\gs^1}[\theta\mid\xi^0]<\infty$. Thus, $F$ rules out an event that occurs in uniformly bounded expected time under every normal-type strategy. 

I verify that perfect falsification implies the hypotheses of \cref{prop:genfalsification}. Suppose $F$ is perfectly falsifiable, and let $C:=\sup_{\gs^1}\bE_{\gs^1}[\theta\mid\xi^0]<\infty$. For every $K>0$, choose $m_{\gd,K}:=\lceil(1-\gd)^{-1/2}\rceil$. Then $(1-\gd)m_{\gd,K}\ra0$, and Markov's inequality gives $\inf_{\gs^1}L^{\gs^1}(\theta\leq m_{\gd,K})\geq1-C/m_{\gd,K}\ra1$. Because $\theta$ is a stopping time, $\{\theta\leq m_{\gd,K}\}$ is a union of length-$m_{\gd,K}$ histories. The union has zero $F$-probability, so each of these histories has zero probability under $F$. Full-support monitoring gives each history positive probability under every normal-type strategy. The log-likelihood ratio is therefore $\infty$ on this event, proving \eqref{eq:early-separation-likelihood}. 

\subsection{Omitted calculations from \cref{rem:diffusebayesianuncertainty}}
\label{sec:appendix-diffusebayesianuncertainty}

Here I prove that the entropy rate $\overline D(\widehat P\parallel F)$ is zero in the example in \cref{rem:diffusebayesianuncertainty}. Fix $\ve>0$, and let $B_\ve:=\{f:\KL{\widehat\rho}{f}<\ve\}$ and $a_\ve:=\pi(B_\ve)$. Since $f\mapsto\KL{\widehat\rho}{f}$ is continuous at $\widehat\rho$ and equals zero there, $B_\ve$ contains a neighborhood of $\widehat\rho$, so $a_\ve>0$. Let $G^\ve$ be the process induced by conditioning the prior on $B_\ve$, so that $G_T^\ve(h_T)=a_\ve^{-1}\int_{B_\ve}f^{\otimes T}(h_T)\pi(\diff f)$ for every history $h_T$. Then $F_T(h_T)\geq a_\ve G_T^\ve(h_T)$ at every history $h_T$. Since $\pi(\diff f)/a_\ve$ is a probability measure on $B_\ve$, Jensen's inequality for $-\log$ gives, at each history $h_T$,
\begin{align*}
\log\frac{\widehat P_T(h_T)}{G_T^\ve(h_T)}
\leq\frac{1}{a_\ve}\int_{B_\ve}\log\frac{\widehat P_T(h_T)}{f^{\otimes T}(h_T)}\pi(\diff f).
\end{align*}
Taking expectations under $\widehat P_T$ therefore yields
\begin{align*}
\KL{\widehat P_T}{F_T}
\leq-\log a_\ve+\KL{\widehat P_T}{G_T^\ve} &\leq-\log a_\ve+\frac{1}{a_\ve}\int_{B_\ve}\KL{\widehat P_T}{f^{\otimes T}}\pi(\diff f)\\
&=-\log a_\ve+\frac{T}{a_\ve}\int_{B_\ve}\KL{\widehat\rho}{f}\pi(\diff f)\\
&\leq-\log a_\ve+T\ve.
\end{align*}
The equality uses additivity of relative entropy across the $T$ independent signals. The final inequality follows from $\KL{\widehat\rho}{f}<\ve$ on $B_\ve$ and $\pi(B_\ve)=a_\ve$. The above inequalities imply $\overline D(\widehat P\parallel F)\leq\ve$. Since $\ve>0$ is arbitrary, $\overline D(\widehat P\parallel F) = 0$.

\subsection{General monitoring}
\label{sec:slp-dependent-monitoring}

I extend the main results to signals whose distribution depends on both players' actions. Let $\rho(\cdot\mid a,b)\in\Delta(Y)$, with $\rho(y\mid a,b)>0$ for every $(a,b,y)$. The information structure is unchanged. The short-lived player's realized payoff remains $\tilde v(b,y)$, so her ex ante payoff is $v(a,b):=\sum_y\rho(y\mid a,b)\tilde v(b,y)$. For mixed actions $\ga\in\Delta(A)$ and $\gb\in\Delta(B)$, write $\rho_{\ga,b}:=\sum_a\ga(a)\rho(\cdot\mid a,b)$ and $\rho_{\ga,\gb}:=\sum_b\gb(b)\rho_{\ga,b}$. Let $\widehat\rho_b:=\rho_{\widehat\ga,b}$ and $\widehat\rho_\gb:=\rho_{\widehat\ga,\gb}$.

The short-lived player now has a signal prediction for each of her actions. Write $\mathscr Q_B:=\prod_{b\in B}\Delta(Y)$ and $q=(q^b)_{b\in B}\in\mathscr Q_B$, where $q^b$ is her predicted signal distribution if she chooses $b$. For a mixed action $\gb$, let $q^\gb:=\sum_b\gb(b)q^b$. Her best-response set is $\mathcal B(q):=\arg\max_{\gb'\in\Delta(B)}\sum_b\gb'(b)\sum_yq^b(y)\tilde v(b,y)$.

A subjective commitment model is a family $F=(F_t(\cdot\mid h_t,b))_{t,h_t,b}$ of conditional signal distributions in $\Delta(Y)$. Let $\mathcal F_B$ denote the set of these families, and extend each distribution to mixed actions by $F_t(\cdot\mid h_t,\gb):=\sum_b\gb(b)F_t(\cdot\mid h_t,b)$. The correctly specified model $\widehat F$ satisfies $\widehat F_t(\cdot\mid h_t,b)=\widehat\rho_b$ at every history and action.

Let $\Sigma^2$ denote the set of behavior strategies of the short-lived players. For a fixed $\gs^2\in\Sigma^2$, persistent play of $\widehat\ga$ generates the true signal process $\widehat P^{\gs^2}$, while the subjective model induces $F^{\gs^2}$. For $h_T=(y_0,\ldots,y_{T-1})$, their marginals are $\widehat P_T^{\gs^2}(h_T):=\prod_{s=0}^{T-1}\widehat\rho_{\gs_s^2(h_s)}(y_s)$ and $F_T^{\gs^2}(h_T):=\prod_{s=0}^{T-1}F_s(y_s\mid h_s,\gs_s^2(h_s))$.

For $F,G\in\mathcal F_B$ and $T\geq1$, define $d_T(F,G):=\sup_{\gs^2\in\Sigma^2}\|F_T^{\gs^2}-G_T^{\gs^2}\|_{\mathrm{TV}}$. The family of pseudometrics $(d_T)_{T\geq1}$ generates the finite-history topology: $F^n\ra F$ if and only if $d_T(F^n,F)\ra0$ for every fixed $T$. Thus convergence of finite-history distributions is uniform over the short-lived players' strategies at each fixed horizon.

For a strategy profile $\gs$, the normal type's history distribution is $L_T^\gs(h_T):=\prod_{s=0}^{T-1}\rho_{\gs_s^1(h_s,\xi^0),\gs_s^2(h_s)}(y_s)$. The posterior probability of commitment is therefore $\mu_t^{\gs,F}(h_t):=\mu_0F_t^{\gs^2}(h_t)/[\mu_0F_t^{\gs^2}(h_t)+(1-\mu_0)L_t^\gs(h_t)]$. If the current short-lived player chooses $b$, her predictive signal distribution is $q_{t,b}^{\gs,F}(\cdot\mid h_t):=\mu_t^{\gs,F}(h_t)F_t(\cdot\mid h_t,b)+(1-\mu_t^{\gs,F}(h_t))\rho_{\gs_t^1(h_t,\xi^0),b}$. Write $\mathbf q_t^{\gs,F}(h_t):=(q_{t,b}^{\gs,F}(\cdot\mid h_t))_{b\in B}$.

\begin{def*}
A strategy profile $\gs$ is an equilibrium if the normal type maximizes his expected discounted payoff against $\gs^2$ and $\gs_t^2(h_t)\in\mathcal B(\mathbf q_t^{\gs,F}(h_t))$ at every history $h_t$.
\end{def*}

Define $\Gs(\gd,F)$, $U(\gs;\gd,F)$, $\underline W(\gd;F)$, and $\overline W(\gd;F)$ as before, and let $\overline W^{CI}$ denote the complete-information benchmark for this game. The same finite-truncation argument establishes equilibrium existence.

For $\ve\geq0$, let $\BR^2_\ve(\ga)$ be the set of responses $\gb\in\Delta(B)$ for which some $q\in\mathscr Q_B$ satisfies $\gb\in\mathcal B(q)$ and $D(\rho_{\ga,\gb}\parallel q^\gb)\leq\ve$. The entropy bound compares the true and predicted signal distributions when the short-lived player uses $\gb$. 

The analogue of \cref{thm:fl} holds, with the $0$-entropy-confirming commitment payoff $R^c:=\min_{\gb\in\BR^2_0(\widehat\ga)}u(\widehat\ga,\gb)$ replacing $R$ as the canonical reputation guarantee (see, \egg \citealp{gossner2011simple}). Define the collapse, canonical, and survival sets by $\mathcal F_B^{\mathrm{collapse}}:=\{F\in\mathcal F_B:\limsup_{\gd\ra1}\overline W(\gd;F)\leq\overline W^{CI}\}$, $\mathcal F_B^{\mathrm{canonical}}:=\{F\in\mathcal F_B:\liminf_{\gd\ra1}\underline W(\gd;F)\geq R^c\}$, and $\mathcal F_B^{\mathrm{survival}}:=\{F\in\mathcal F_B:\liminf_{\gd\ra1}\underline W(\gd;F)>\overline W^{CI}\}$.

\cref{thm:controlled-fragility} is a direct analogue of my three main results:

\begin{thm*}\label{thm:controlled-fragility}
The following hold.
\begin{enumerate}\itemsep0em
\item Collapse: The set $\mathcal F_B^{\mathrm{collapse}}$ is dense in the finite-history topology. In particular, there exists a sequence $(F^n)_{n=1}^\infty$ of full-support processes along which misspecification vanishes empirically and $\limsup_{\gd\ra1}\overline W(\gd;F^n)\leq\overline W^{CI}$ for every $n$. The sequence can be chosen so that $(F^n)_T^{\gs^2}=\widehat P_T^{\gs^2}$ for every $\gs^2\in\Sigma^2$, every $n\geq1$, and every $T=1,\ldots,n$.
\item Fragility to subjective inference: If $R^c>\overline W^{CI}$, both $\mathcal F_B^{\mathrm{canonical}}$ and $\mathcal F_B^{\mathrm{survival}}$ are dense, and the maps $F\mapsto\limsup_{\gd\ra1}\overline W(\gd;F)$ and $F\mapsto\liminf_{\gd\ra1}\underline W(\gd;F)$ are nowhere continuous.
\item Fragility to discounting: The set
\begin{align}
\mathcal G_B
:=\left\{F\in\mathcal F_B:
\liminf_{\gd\ra1}\overline W(\gd;F)\leq\overline W^{CI}
\text{ and }
\limsup_{\gd\ra1}\underline W(\gd;F)\geq R^c\right\}
\label{eq:controlled-generic-instability}
\end{align}
is comeagre in the finite-history topology. Therefore, if $R^c>\overline W^{CI}$, then $\mathcal F_B^{\mathrm{collapse}}$ and $\mathcal F_B^{\mathrm{survival}}$ are meagre.
\end{enumerate}
\end{thm*}

The proof is in \cref{sec:proofcontrolledextension}.

\section{Proofs}
\label{sec:proofs}

\subsection{Regularity}
\label{sec:regularity}

In this subsection I establish a useful regularity property of equilibrium payoffs.

\begin{lem*}[Semicontinuity]
\label{lem:payoff-semicontinuity}
The maps $(\gd,F)\mapsto\underline W(\gd;F)$ and $(\gd,F)\mapsto\overline W(\gd;F)$ are, respectively, lower and upper semicontinuous on $(0,1)\times\Delta(H)$. 
\end{lem*}

\begin{proof}[Proof of \cref{lem:payoff-semicontinuity}]
Give the strategy-profile space $\mathscr S$ its product topology. As a countable product of finite-dimensional simplexes, it is compact metrizable. Fixing $F$, write $U_\gd(\gs)$ for the normal type's payoff. Its truncation to dates $0,\ldots,T-1$ is continuous in $(\gd,\gs)$, and the omitted tail is bounded by $M\gd^T$, where $M:=\max_{a,b}|u(a,b)|$. This bound tends to zero uniformly over strategy profiles and discount factors in any compact subset of $(0,1)$. Thus $U_\gd(\gs)$ is jointly continuous in $(\gd,\gs)$.

I next show that the equilibrium graph is closed. For each normal-type deviation $\tau^1$, the optimality inequality $U_\gd(\gs)\geq U_\gd(\tau^1,\gs^2)$ is a closed condition in $(\gd,F,\gs)$ by payoff continuity. For the short-lived players, the predictive probability at any fixed history $h_t$ is 
\begin{align}
q_t^{\gs^1}(y\mid h_t)
=\frac{\mu_0F_{t+1}(h_ty)
+(1-\mu_0)L_t^{\gs^1}(h_t)\rho_{\gs_t^1(h_t,\xi^0)}(y)}
{\mu_0F_t(h_t)+(1-\mu_0)L_t^{\gs^1}(h_t)},
\label{eq:continuous-predictive-belief}
\end{align}
where $h_ty$ denotes $h_t$ followed by $y$. The finite-history probabilities in this expression are continuous in $(F,\gs^1)$. Moreover, the denominator is at least $(1-\mu_0)\underline\rho^t>0$, where $\underline\rho:=\min_{a,y}\rho(y\mid a)$. Hence predictions are continuous even at $F$-null histories. The closed graph of $\BR^2$ then makes each short-lived player's optimality condition closed. Intersecting these conditions over all deviations and histories proves that the equilibrium graph is closed.

The set of equilibria is nonempty by the truncation argument in \cref{sec:model} and compact by closedness of the equilibrium graph and compactness of $\mathscr S$. Payoff continuity then implies that both equilibrium-payoff bounds $\underline W$ and $\overline W$ are attained.

Finally, let $(\gd_n,F^n)\ra(\gd,F)$. Choose a subsequence along which $\underline W(\gd_n;F^n)$ converges to $\liminf_n\underline W(\gd_n;F^n)$, and at each term $(\gd_n, F^n)$ along the sequence choose a minimizing equilibrium $\gs^n$ attaining $\underline W(\gd_n;F^n)$. Compactness of $\mathscr S$ gives a further subsequence of these equilibria converging to some $\gs^*$. Closedness of the equilibrium graph implies that $\gs^*\in\Gs(\gd,F)$, and payoff continuity gives $\underline W(\gd;F)\leq U_\gd(\gs^*)=\liminf_n\underline W(\gd_n;F^n)$. The same argument with maximizing equilibria along a subsequence realizing the limit superior gives $\limsup_n\overline W(\gd_n;F^n)\leq\overline W(\gd;F)$. These are the claimed semicontinuity properties.
\end{proof}

\subsection{Proof of \cref{prop:patience-hurts}}

The canonical bound \eqref{eq:FLbound} and $R>\overline W^{CI}$ give some $\gd_0\in(0,1)$ with $\underline W(\gd_0;\widehat P)>\overline W^{CI}$. Since $F^n\ra\widehat P$, lower semicontinuity of the map $F\mapsto\underline W(\gd_0;F)$, established in \cref{lem:payoff-semicontinuity} in \cref{sec:regularity}, gives $\liminf_{n\to\infty}\underline W(\gd_0;F^n)\geq\underline W(\gd_0;\widehat P)>\overline W^{CI}$. Thus the first inequality in \eqref{eq:patience-hurts} holds for every sufficiently large $n$. The second holds for every $n$ because $F^n\in\mathcal F^{\mathrm{collapse}}$.

\subsection{Proof of \cref{prop:genfalsification}}

I first establish a complete-information payoff bound allowing approximate best responses by the short-lived players, and then apply it when their posterior probability of commitment is small. Let $\underline u:=\min_{a,b}u(a,b)$ and $\overline u:=\max_{a,b}u(a,b)$.

\begin{lem*}\label{lem:tciub}
For every $\eta>0$, there exist $\ve>0$ and $L>0$ such that the following holds for every $\gd\in(0,1)$. Let $N_\gd:=\lceil L/(1-\gd)\rceil$. Consider a strategy profile of the complete-information game with horizon $N_\gd$ and terminal continuation payoffs in $[\underline u,\overline u]$. Suppose that player 1 is sequentially optimal and that, at every nonterminal history, the short-lived player's best-response loss against player 1's mixed action is at most $\ve$. Then player 1's initial payoff is at most $\overline W^{CI}+\eta$.
\end{lem*}

\begin{proof}[Proof of \cref{lem:tciub}]
For mixed actions $(\alpha,\gb)$, define the short-lived player's best-response loss by $\ell_2(\alpha,\gb):=\max_{\gb'\in\Delta(B)}\bE_{\rho_\alpha}[\tilde v(\gb',y)]-\bE_{\rho_\alpha}[\tilde v(\gb,y)]$. For nonnegative signal-dependent penalties $z\in\bR_+^Y$, define player 1's score by $S(a;\gb,z):=u(a,\gb)-\sum_y\rho(y\mid a)z(y)$. For $\ve\geq0$, let
\begin{align}
\begin{multlined}[15mm]
\kappa_\ve:=\sup\Bigl\{r: \alpha\in\Delta(A), \gb\in\Delta(B), z\in\bR_+^Y,
\\ \ell_2(\alpha,\gb)\leq\ve, \quad r=\max_{a\in A}S(a;\gb,z), \quad \operatorname{supp}(\alpha)\subseteq\arg\max_{a\in A}S(a;\gb,z)
\Bigr\}.
\end{multlined}
\label{eq:claim1-kappa}
\end{align}
Thus $\kappa_\ve$ is the highest feasible score when player 1 maximizes his score and the short-lived player's best-response loss is at most $\ve$. A one-shot Nash equilibrium with $z=0$ is feasible for $\kappa_0$, and nonnegative penalties make every score at most $\overline u$. Moreover, increasing $\ve$ enlarges the feasible set. Hence $\underline u\leq\kappa_0\leq\kappa_\ve\leq\overline u$.

\paragraph{Step 1.} I establish $\kappa_0=\overline W^{CI}$. Fix $\gd\in(0,1)$. The upper maximal-score problem of \citet{fudenberg1994efficiency} maximizes $r$ over mixed actions $(\alpha,\gb)$ and continuation payoffs $w(y)\leq r$, subject to $\ell_2(\alpha,\gb)=0$ and
\begin{align}
r&=(1-\gd)u(a,\gb)+\gd\sum_y\rho(y\mid a)w(y)
&&\text{for }a\in\operatorname{supp}(\alpha),
\label{eq:claim1-score-equality}\\
r&\geq(1-\gd)u(a,\gb)+\gd\sum_y\rho(y\mid a)w(y)
&&\text{for }a\in A.
\label{eq:claim1-score-inequality}
\end{align}
Set $z(y):=\gd(r-w(y))/(1-\gd)\geq0$. Substitution into \eqref{eq:claim1-score-equality} and \eqref{eq:claim1-score-inequality} gives $r=S(a;\gb,z)$ for any $a\in\operatorname{supp}(\alpha)$ and $r\geq S(a;\gb,z)$ for any $a \in A$. These are the conditions defining $\kappa_0$ in \eqref{eq:claim1-kappa}. Conversely, any feasible tuple for $\kappa_0$ yields feasible continuation payoffs by setting $w(y)=r-(1-\gd)z(y)/\gd$. Thus the upper maximal-score bound is $\kappa_0$.

The corresponding lower-bound problem minimizes $r$ under the same incentive conditions but requires $w(y)\geq r$. Let $Q^{CI}$ be the closed interval between the lower and upper bounds. The lower bound is at least $\underline u$, and the upper bound is at most $\overline u$. A one-shot Nash equilibrium with constant continuation payoffs is feasible for both problems, so $Q^{CI}$ is nonempty and compact.

Full-support monitoring makes every finite public history occur with positive probability. Any profitable continuation deviation would therefore yield a profitable deviation from the outset. Thus the complete-information Nash and perfect public equilibrium payoff sets coincide. By \citet[Theorem 3.1(i)]{fudenberg1994efficiency}, these payoffs lie in $Q^{CI}$, giving $\overline W^{CI}\leq\kappa_0$. If $Q^{CI}$ has positive length, it is full-dimensional in the one-dimensional payoff space. Their Theorem 3.1(ii) then implies that equilibrium payoffs approach its upper endpoint $\kappa_0$ as $\gd\ra1$. If $Q^{CI}$ is a singleton, equilibrium existence and containment imply that every equilibrium payoff equals $\kappa_0$. Thus $\kappa_0=\overline W^{CI}$ in either case.

\paragraph{Step 2.} I show that
\begin{align}
\lim_{\ve\searrow0}\kappa_\ve=\overline W^{CI}.
\label{eq:claim1-kappa-limit}
\end{align}
Since $\kappa_\ve$ decreases as $\ve\searrow0$ and is bounded below by $\kappa_0$, its limit exists. Take $\ve_m\searrow0$ and choose feasible tuples $(\alpha_m,\gb_m,z_m,r_m)$ satisfying $\kappa_{\ve_m}-1/m\leq r_m\leq\kappa_{\ve_m}$. Then $r_m$ converges to that limit and lies in $[\underline u-1,\overline u]$.

The support condition gives $r_m=u(\alpha_m,\gb_m)-\sum_y\rho_{\alpha_m}(y)z_m(y)$. Let $\underline\rho:=\min_{a,y}\rho(y\mid a)>0$. Since $\rho_{\alpha_m}(y)\geq\underline\rho$ for each $y$, $\underline\rho\sum_y z_m(y)\leq\overline u-\underline u+1$. Thus the penalty vectors $z_m=(z_m(y))_{y\in Y}$ are bounded. Compactness of the mixed-action simplexes  gives a subsequence along which $(\alpha_m,\gb_m,z_m,r_m)\ra(\alpha,\gb,z,r)$, with $z\geq0$ and $r=\lim_{\ve\searrow0}\kappa_\ve$.

Continuity of $\ell_2$ in $(\alpha,\gb)$ and $0\leq\ell_2(\alpha_m,\gb_m)\leq\ve_m\ra0$ give $\ell_2(\alpha,\gb)=0$. For each action $a$, $S(a;\gb,z)$ is continuous in $(\gb,z)$. Since $A$ is finite, taking limits in $r_m=\max_aS(a;\gb_m,z_m)$ gives $r=\max_aS(a;\gb,z)$. Moreover, if $\alpha(a)>0$, then $\alpha_m(a)>0$ for all sufficiently large $m$, so feasibility gives $S(a;\gb_m,z_m)=r_m$. Taking limits in this equality yields $S(a;\gb,z)=r$. Thus every action in the support of $\alpha$ maximizes the score, and the limit tuple is feasible for $\kappa_0$. Hence $r\leq\kappa_0$. Since $r=\lim_{\ve\searrow0}\kappa_\ve\geq\kappa_0$, equality follows. Together with Step 1, this proves \eqref{eq:claim1-kappa-limit}.

\paragraph{Step 3.} Fix $\ve\geq0$, $\gd\in(0,1)$, and a finite horizon $N$. Consider a strategy profile in which player 1 is sequentially optimal and each short-lived player's best-response loss is at most $\ve$ at every nonterminal history. Write $U_s(h_s)$ for player 1's normalized continuation payoff conditional on history $h_s$ and let $M_s:=\max_{h_s\in Y^s}U_s(h_s)$. The terminal-payoff assumption gives $M_N\leq\overline u$.

Fix a history $h_s$ with $s<N$, and write $\alpha:=\gs_s^1(h_s,\xi^0)$, $\gb:=\gs_s^2(h_s)$, and $w(y):=U_{s+1}(h_sy)$. Sequential optimality gives $U_s(h_s)=\max_a\{(1-\gd)u(a,\gb)+\gd\sum_y\rho(y\mid a)w(y)\}$, with every action in the support of $\alpha$ attaining the maximum. Since $w(y)\leq M_{s+1}$, define $z(y):=\gd(M_{s+1}-w(y))/(1-\gd)\geq0$. Substituting for $w(y)$ in the preceding equality gives $U_s(h_s)=\gd M_{s+1}+(1-\gd)\max_aS(a;\gb,z)$. Every action in the support of $\alpha$ also maximizes this score. Together with $\ell_2(\alpha,\gb)\leq\ve$, this makes the maximum score feasible for $\kappa_\ve$. Hence $U_s(h_s)\leq(1-\gd)\kappa_\ve+\gd M_{s+1}$.

Taking the maximum over histories gives $M_s\leq(1-\gd)\kappa_\ve+\gd M_{s+1}$. Applying this inequality successively from date $0$ to date $N-1$, and using $M_N\leq\overline u$ and $\kappa_\ve\geq\underline u$, yields
\begin{align}
M_0
&\leq(1-\gd^N)\kappa_\ve+\gd^N M_N
\leq\kappa_\ve+\gd^N(\overline u-\underline u).
\label{eq:claim1-finite-bound}
\end{align}
Given $\eta>0$, choose $\ve>0$ using \eqref{eq:claim1-kappa-limit} so that $\kappa_\ve\leq\overline W^{CI}+\eta/2$, and choose $L>0$ so that $\expo^{-L}(\overline u-\underline u)\leq\eta/2$. Both choices are independent of $\gd$. For $N_\gd=\lceil L/(1-\gd)\rceil$, the inequality $\log\gd\leq-(1-\gd)$ gives $\gd^{N_\gd}\leq\expo^{-(1-\gd)N_\gd}\leq\expo^{-L}$. Applying \eqref{eq:claim1-finite-bound} with $N=N_\gd$ therefore gives $M_0\leq\overline W^{CI}+\eta$, as required.
\end{proof}

Define the posterior odds of commitment against the normal type by $O_t^{\gs^1}(h_t):=\mu_t^{\gs^1}(h_t)/(1-\mu_t^{\gs^1}(h_t))=\mu_0F_t(h_t)/[(1-\mu_0)L_t^{\gs^1}(h_t)]$.

\begin{lem*}
\label{lem:small-posterior-ci}
For every $\eta>0$, there exist $K>0$ and $\bar\gd\in(0,1)$ such that, for every $F\in\Delta(H)$, $\gd\in(\bar\gd,1)$, $\gs\in\Gs(\gd,F)$, and history $h_t$, if $O_t^{\gs^1}(h_t)\leq\expo^{-K/(1-\gd)}$, then the normal type's continuation payoff after $h_t$ is at most $\overline W^{CI}+\eta$.
\end{lem*}

\begin{proof}[Proof of \cref{lem:small-posterior-ci}]
Fix $\eta>0$ and choose $\ve,L>0$ from \cref{lem:tciub}. Let $\underline\rho:=\min_{a,y}\rho(y\mid a)>0$, $V:=\max\{1,\|\tilde v\|_\infty\}$, and $N_\gd:=\lceil L/(1-\gd)\rceil$. Choose $K>L|\log\underline\rho|$. Since $N_\gd\leq L/(1-\gd)+1$, the exponent $-K/(1-\gd)+N_\gd|\log\underline\rho|$ tends to $-\infty$ as $\gd\ra1$. Choose $\bar\gd\in(0,1)$ so that, for every $\gd>\bar\gd$,
\begin{align}
\expo^{-K/(1-\gd)+N_\gd|\log\underline\rho|}
\leq\frac{\ve}{4V}.
\label{eq:lemma2-window}
\end{align}
Fix $F\in\Delta(H)$, $\gd\in(\bar\gd,1)$, $\gs\in\Gs(\gd,F)$, and a history $h_t$ satisfying $O_t^{\gs^1}(h_t)\leq\expo^{-K/(1-\gd)}$. For any sequence $z_s\in Y^s$ of subsequent signals, $F_{t+s}(h_tz_s)\leq F_t(h_t)$, whereas $L_{t+s}^{\gs^1}(h_tz_s)\geq L_t^{\gs^1}(h_t)\underline\rho^s$ because every signal has conditional probability at least $\underline\rho$ under the normal type. The formula for posterior odds therefore gives
\begin{align}
O_{t+s}^{\gs^1}(h_tz_s)
\leq O_t^{\gs^1}(h_t)\underline\rho^{-s}.
\label{eq:lemma2-odds-growth}
\end{align}
Combining this inequality with \eqref{eq:lemma2-window} yields, for every $s\leq N_\gd$ and continuation $z_s$,
\begin{align}
\mu_{t+s}^{\gs^1}(h_tz_s)
\leq O_{t+s}^{\gs^1}(h_tz_s)
\leq\frac{\ve}{4V}.
\label{eq:lemma2-future-belief}
\end{align}
At any such history with $s<N_\gd$, write $q^0:=\rho_{\gs_{t+s}^1(h_tz_s,\xi^0)}$ for the normal-type signal distribution. The short-lived player's prediction is $q=\mu q^c+(1-\mu)q^0$, where $\mu$ is her posterior probability of the commitment type and $q^c$ is her prediction conditional on the commitment type. When $\mu=0$, choose $q^c$ arbitrarily. Then
\begin{align}
\|q-q^0\|_{\mathrm{TV}}
=\mu\|q^c-q^0\|_{\mathrm{TV}}
\leq\mu\leq\frac{\ve}{4V}.
\label{eq:lemma2-tv}
\end{align}
Let $\gb$ be her equilibrium mixed action at that history. Changing the signal distribution from $q$ to $q^0$ changes the expected payoff of any mixed action by at most $2V\|q-q^0\|_{\mathrm{TV}}$. Since $\gb$ is a best response to $q$, every alternative $\gb'$ therefore satisfies $\bE_{q^0}[\tilde v(\gb',y)-\tilde v(\gb,y)]\leq4V\|q-q^0\|_{\mathrm{TV}}\leq\ve$. Thus her best-response loss against the normal-type prediction is at most $\ve$ at every history during the next $N_\gd$ periods.

Starting from $h_t$, truncate the equilibrium strategy profile after $N_\gd$ periods and assign the normal type's actual normalized equilibrium continuation payoff at each terminal history. These terminal payoffs lie in $[\underline u,\overline u]$, and the truncation preserves his initial continuation payoff. He remains sequentially optimal: any profitable deviation in the truncated game, followed by the assigned equilibrium continuation, would yield a profitable continuation deviation in the original game. Full-support monitoring gives every finite public history positive probability conditional on the normal type, so such a deviation would contradict equilibrium. Each short-lived player's best-response loss against the normal type is at most $\ve$, as established above. Applying \cref{lem:tciub} to this complete-information strategy profile gives the required bound $\overline W^{CI}+\eta$.
\end{proof}

I now complete the proof of \cref{prop:genfalsification}. Fix $F$ satisfying its hypotheses and $\eta>0$. Choose $K>0$ and $\bar\gd<1$ from \cref{lem:small-posterior-ci}, and apply \eqref{eq:early-separation-date} and \eqref{eq:early-separation-likelihood} with $2K$. Write $m_\gd:=m_{\gd,2K}$. For each normal-type strategy $\gs^1$, let $G_\gd^{\gs^1}$ be the event that $\log(L_{m_\gd}^{\gs^1}(h_{m_\gd})/F_{m_\gd}(h_{m_\gd}))\geq2K/(1-\gd)$, interpreting the ratio as infinity when its denominator is zero.

Fix an equilibrium $\gs\in\Gs(\gd,F)$. On $G_\gd^{\gs^1}$, the posterior odds satisfy $O_{m_\gd}^{\gs^1}\leq[\mu_0/(1-\mu_0)]\expo^{-2K/(1-\gd)}\leq\expo^{-K/(1-\gd)}$ for all $\gd$ sufficiently close to one. Because the prior odds are fixed, this discount-factor threshold is independent of the equilibrium. Let $p:=L^{\gs^1}((G_\gd^{\gs^1})^c)$. Payoffs before date $m_\gd$ have total discounted weight $1-\gd^{m_\gd}$ and are bounded above by $\overline u$. The normal type's continuation payoff at date $m_\gd$ is bounded above by $\overline W^{CI}+\eta$ on $G_\gd^{\gs^1}$ by \cref{lem:small-posterior-ci} and by $\overline u$ otherwise. Hence
\begin{align}
U(\gs;\gd,F)
&\leq(1-\gd^{m_\gd})\overline u
+\gd^{m_\gd}\bigl[(1-p)(\overline W^{CI}+\eta)+p\overline u\bigr]
\nonumber\\
&\leq\overline W^{CI}+\eta
+(\overline u-\underline u)\bigl[1-\gd^{m_\gd}+p\bigr],
\label{eq:genfalsification-payoff-bound}
\end{align}
where the second inequality uses $\overline W^{CI}+\eta\geq\underline u$.

Both terms $1-\gd^{m_\gd}$ and $p$ in \eqref{eq:genfalsification-payoff-bound} vanish as $\gd \ra 1$ uniformly over equilibria: $1-\gd^{m_\gd}\leq(1-\gd)m_\gd\ra0$ by \eqref{eq:early-separation-date}, and $p\leq\sup_{\gs^1}L^{\gs^1}((G_\gd^{\gs^1})^c)\ra0$ by \eqref{eq:early-separation-likelihood}. Taking the supremum over equilibria in \eqref{eq:genfalsification-payoff-bound}, followed by the limit superior as $\gd\ra1$, gives $\limsup_{\gd\ra1}\overline W(\gd;F)\leq\overline W^{CI}+\eta$. Letting $\eta\searrow0$ proves $F\in\mathcal F^{\mathrm{collapse}}$.

\subsection{Proof of \cref{thm:fd-collapse}}

Fix $G\in\Delta(H)$ and a finite horizon $T$. I construct $F\in\mathcal F^{\mathrm{collapse}}$ that agrees with $G$ from time $0$ through horizon $T$. Choose $\widetilde y\in Y$, let $r:=\widehat\rho(\widetilde y)\in(0,1)$, and set $F_T:=G_T$. For every history $h_t$ with $t\geq T$, define
\begin{align}
F_t(y\mid h_t)
:=
\begin{cases}
r\expo^{-(t-T)^2}, & y=\widetilde y,\\[0.5em]
\dfrac{1-r\expo^{-(t-T)^2}}{1-r}\widehat\rho(y), & y\neq\widetilde y.
\end{cases}
\label{eq:full-support-collapse-tail}
\end{align}
All these conditional probabilities are positive, so $F$ has full support whenever $G_T$ does.

I verify the hypotheses of \cref{prop:genfalsification}, which gives $F\in\mathcal F^{\mathrm{collapse}}$. Let $\underline\rho:=\min_{a,y}\rho(y\mid a)>0$ and $\underline\rho_{\widetilde y}:=\min_a\rho(\widetilde y\mid a)>0$. Fix $K>0$, set $\ell_\gd:=\lceil\sqrt{4K/(1-\gd)}\rceil$, and choose $m_{\gd,K}:=T+2\ell_\gd+1$. Then $(1-\gd)m_{\gd,K}\ra0$, as required by \eqref{eq:early-separation-date}. Next, let $E_\gd$ be the event that $\widetilde y$ occurs at least once at dates $T+\ell_\gd,\ldots,T+2\ell_\gd$. These $\ell_\gd+1$ signals are included in the history at date $m_{\gd,K}$. Under every normal-type strategy, the conditional probability of $\widetilde y$ at each date is at least $\underline\rho_{\widetilde y}$. Successive conditioning therefore gives $\inf_{\gs^1}L^{\gs^1}(E_\gd)\geq1-(1-\underline\rho_{\widetilde y})^{\ell_\gd+1}\ra1$. On $E_\gd$, the signal $\widetilde y$ occurs at some date $t\geq T+\ell_\gd$, when its subjective conditional probability is at most $r\expo^{-\ell_\gd^2}\leq\expo^{-\ell_\gd^2}$. Every other factor in the subjective likelihood, including the initial probability $G_T(h_T)$, is at most one. Thus $F_{m_{\gd,K}}(h_{m_{\gd,K}})\leq\expo^{-\ell_\gd^2}$. Meanwhile, full-support monitoring gives $L_{m_{\gd,K}}^{\gs^1}(h_{m_{\gd,K}})\geq\underline\rho^{m_{\gd,K}}$ under every normal-type strategy. Consequently, on $E_\gd$, $\log(L_{m_{\gd,K}}^{\gs^1}/F_{m_{\gd,K}})\geq\ell_\gd^2-m_{\gd,K}|\log\underline\rho|\geq K/(1-\gd)$ for all $\gd$ sufficiently close to one. The last inequality follows from $(1-\gd)\ell_\gd^2\geq4K$ and $(1-\gd)m_{\gd,K}\ra0$, with a threshold independent of $\gs^1$. If the subjective likelihood is zero, the log-likelihood ratio is infinite. Since the probability of $E_\gd$ tends to one uniformly over normal-type strategies, \eqref{eq:early-separation-likelihood} follows.

It remains to prove density and \eqref{eq:collapse}. For each $n\geq1$, apply this construction with $T=n$ to obtain $F^n\in\mathcal F^{\mathrm{collapse}}$. For every fixed horizon $T$, $F_T^n=G_T$ whenever $n\geq T$. Thus $F^n\ra G$ in the finite-history topology, proving density. Taking $G=\widehat P$ gives full-support processes whose misspecification vanishes empirically and for which, for every $n$, \eqref{eq:collapse} holds and $F^n_T = \widehat P_T$ for every $T=1,\dots,n$.

\subsection{Proof of \cref{prop:standard-misspecifications}}

Since $F_T(h_T)\geq\lambda\widehat P_T(h_T)$ for every $T$ and $h_T$, taking logarithms of the likelihood ratio and expectations under $\widehat P_T$ gives $0\leq\KL{\widehat P_T}{F_T}\leq-\log\lambda$. Dividing by $T$ and letting $T\ra\infty$ yields $\overline D(\widehat P\parallel F)=0$.

\subsection{Proof of \cref{prop:entropy-rate-robustness}}

Extend \eqref{eq:pdentropyrate} to arbitrary $P,F\in\Delta(H)$ by
\begin{align}
D_\gd(P\parallel F)
:=
(1-\gd)\sum_{t=0}^{\infty}\gd^t
\bE_P\!\left[
\KL{P(\cdot\mid h_t)}{F(\cdot\mid h_t)}
\right].
\label{eq:general-discounted-entropy}
\end{align}
By direct computation,
\begin{align}
D_\gd(P\parallel F)
=(1-\gd)^2\sum_{T=1}^{\infty}\gd^{T-1}\KL{P_T}{F_T}.
\label{eq:discounted-chain-rule}
\end{align}
Fix $F\in\Delta(H)$. Under a strategy profile $\gs$, the short-lived players' subjective signal distribution is $Q^{\gs^1}:=\mu_0F+(1-\mu_0)L^{\gs^1}$, whose predictive signal distribution is $q_t^{\gs^1}(\cdot\mid h_t)$ at history $h_t$.

\begin{lem*}
\label{lem:finite-horizon-budget}
For every strategy profile $\gs$ and horizon $T\geq1$,
\begin{align}
\sum_{t=0}^{T-1}\bE_{\widehat P}\!\left[
\KL{\widehat\rho}{q_t^{\gs^1}(\cdot\mid h_t)}
\right]
\leq-\log\mu_0+\KL{\widehat P_T}{F_T}.
\label{eq:finite-history-entropy-budget}
\end{align}
Consequently, \eqref{eq:discounted-entropy-budget} holds for every $\gd\in(0,1)$.
\end{lem*}

\begin{proof}[Proof of \cref{lem:finite-horizon-budget}]
By the chain rule of relative entropy, the left side of \eqref{eq:finite-history-entropy-budget} equals $\KL{\widehat P_T}{Q_T^{\gs^1}}$. Since $Q_T^{\gs^1}\geq\mu_0F_T$ pointwise, taking logarithms of the likelihood ratios and expectations under $\widehat P_T$ gives $\KL{\widehat P_T}{Q_T^{\gs^1}}\leq-\log\mu_0+\KL{\widehat P_T}{F_T}$, proving \eqref{eq:finite-history-entropy-budget}. Multiply this inequality by $(1-\gd)^2\gd^{T-1}$ and sum over $T\geq1$. The weights sum to $1-\gd$, so \eqref{eq:discounted-chain-rule} gives \eqref{eq:discounted-entropy-budget}.
\end{proof}

I next establish an essential payoff bound, given in \eqref{eq:finite-patience-robustness} below. Fix $F\in\Delta(H)$, $\gd\in(0,1)$, $\gs\in\Gs(\gd,F)$, and $w<R$. For $\ga\in\Delta(A)$, let $\underline u_\ga:=\min_{b\in B}u(\ga,b)$ and $\phi_\ga(q):=\min_{\gb\in\BR^2(q)}u(\ga,\gb)$. The latter is the lowest payoff from $\ga$ across best responses to prediction $q$; the minimum exists because $\BR^2(q)$ is nonempty and compact and $u(\ga,\cdot)$ is continuous. Define
\begin{align}
m_\ga(w)
:=
\inf_{ q\in\Delta(Y):\,\phi_\ga(q)<w }
\KL{\rho_\ga}{q},
\label{eq:margin}
\end{align}
with the infimum of the empty set equal to $+\infty$.

For $w<R$, $m_{\widehat\ga}(w)>0$. Otherwise, there would be predictions $q_j$ and best responses $\gb_j\in\BR^2(q_j)$ such that $\KL{\widehat\rho}{q_j}\ra0$ and $u(\widehat\ga,\gb_j)<w$. \citeauthor{pinsker1964information}'s (\citeyear{pinsker1964information}) inequality gives $q_j\ra\widehat\rho$, and compactness of $\Delta(B)$ gives a subsequence with $\gb_j\ra\gb$. The closed graph of $\BR^2$ implies $\gb\in\BR^2(\widehat\rho)$. Continuity of $u(\widehat\ga,\cdot)$ then gives $R\leq u(\widehat\ga,\gb)\leq w$, a contradiction.

Consider the normal type's deviation to persistent play of $\widehat\ga$. Histories follow $\widehat P$, while the short-lived players retain their equilibrium strategies and best respond to $q_t^{\gs^1}(\cdot\mid h_t)$. If $w\leq\underline u_{\widehat\ga}$, the deviation gives at least $w$ in every period. The same holds if $m_{\widehat\ga}(w)=\infty$: every equilibrium prediction has full support and hence finite divergence from $\widehat\rho$, so it cannot admit a best response paying less than $w$. In either case, the desired bound given in \eqref{eq:finite-patience-robustness} below holds with its entire subtracted term defined to be zero, even when $D_\gd(\widehat P\parallel F)=\infty$.

It remains to consider $w>\underline u_{\widehat\ga}$ and $0<m_{\widehat\ga}(w)<\infty$. For any full-support prediction $q$ and any $\gb\in\BR^2(q)$, $u(\widehat\ga,\gb)\geq w-[(w-\underline u_{\widehat\ga})/m_{\widehat\ga}(w)]\KL{\widehat\rho}{q}$. Indeed, if $u(\widehat\ga,\gb)\geq w$, the inequality follows from nonnegativity of relative entropy. Otherwise, $\phi_{\widehat\ga}(q)<w$, so $\KL{\widehat\rho}{q}\geq m_{\widehat\ga}(w)$. The right side is then at most $\underline u_{\widehat\ga}$, which bounds $u(\widehat\ga,\gb)$ from below.

Apply this inequality at each history $h_t$ with $q=q_t^{\gs^1}(\cdot\mid h_t)$ and $\gb=\gs_t^2(h_t)$. Taking expectations under $\widehat P$, summing with weights $(1-\gd)\gd^t$, and applying \eqref{eq:discounted-entropy-budget} bounds the deviation payoff from below. Equilibrium optimality gives the same lower bound for the normal type's equilibrium payoff. Since the bound is independent of equilibrium $\gs$, taking the infimum over equilibria yields
\begin{align}
\underline W(\gd;F)
\geq
w
-
\frac{\max\{0,w-\underline u_{\widehat\ga}\}}
{m_{\widehat\ga}(w)}
\left[
-(1-\gd)\log\mu_0
+D_\gd(\widehat P\parallel F)
\right].
\label{eq:finite-patience-robustness}
\end{align}

\begin{lem*}
\label{lem:entropy-rate-implies-discounted}
For every $F\in\Delta(H)$,
\begin{align}
\limsup_{\gd\ra1}D_\gd(\widehat P\parallel F)
\leq
\overline D(\widehat P\parallel F).
\label{eq:limll}
\end{align}
\end{lem*}

\begin{proof}[Proof of \cref{lem:entropy-rate-implies-discounted}]
Write $d_T:=\KL{\widehat P_T}{F_T}$ and $d:=\limsup_{T\ra\infty}d_T/T$. The claim is immediate if $d=\infty$. Otherwise, every $d_T$ is finite because the chain rule of relative entropy makes $(d_T)_T$ nondecreasing.

Fix $\eta>0$ and choose $T_0$ so that $d_T\leq(d+\eta)T$ for all $T\geq T_0$. Split the sum in \eqref{eq:discounted-chain-rule} at $T_0$. Using $\gd^{T-1}\leq1$ for the initial terms and $(1-\gd)^2\sum_{T\geq1}T\gd^{T-1}=1$ for the remaining terms gives $D_\gd(\widehat P\parallel F)\leq(1-\gd)^2\sum_{T=1}^{T_0-1}d_T+d+\eta$. The finite sum is independent of $\gd$, so $\limsup_{\gd\ra1}D_\gd(\widehat P\parallel F)\leq d+\eta$. Letting $\eta\searrow0$ proves \eqref{eq:limll}.
\end{proof}

I now complete the proof of \cref{prop:entropy-rate-robustness}. Fix $F$ and $w<R$. Taking the limit inferior as $\gd\ra1$ in \eqref{eq:finite-patience-robustness}, using the vanishing prior term and \eqref{eq:limll}, gives
\begin{align}
\liminf_{\gd\ra1}\underline W(\gd;F)
\geq
w
-
\frac{\max\{0,w-\underline u_{\widehat\ga}\}}
{m_{\widehat\ga}(w)}
\overline D(\widehat P\parallel F).
\label{eq:modulus}
\end{align}
As before, the subtracted term is defined to be zero when its coefficient is zero, even if the entropy rate is infinite. Now suppose $\overline D(\widehat P\parallel F^n)\ra0$. For each fixed $w<R$, the coefficient in \eqref{eq:modulus} is finite and independent of $n$. Applying the bound to $F^n$ and taking the limit inferior as $n\ra\infty$ therefore gives $\liminf_{n\ra\infty}\liminf_{\gd\ra1}\underline W(\gd;F^n)\geq w$. Letting $w\nearrow R$ proves \eqref{eq:entropyrobust}.

\subsection{Proof of \cref{thm:nowherecontinuity}}

Fix $G\in\Delta(H)$ and let $F^\lambda:=\lambda\widehat P+(1-\lambda)G$ for $\lambda\in(0,1]$. By \cref{prop:standard-misspecifications}, $\overline D(\widehat P\parallel F^\lambda)=0$. Applying \cref{prop:entropy-rate-robustness} to the constant sequence $F^n=F^\lambda$ gives $\liminf_{\gd\ra1}\underline W(\gd;F^\lambda)\geq R$. Thus $F^\lambda\in\mathcal F^{\mathrm{canonical}}\subseteq\mathcal F^{\mathrm{survival}}$, since $R>\overline W^{CI}$. For every horizon $T$, $\|F_T^\lambda-G_T\|_{\mathrm{TV}}=\lambda\|\widehat P_T-G_T\|_{\mathrm{TV}}\leq\lambda$, so $F^\lambda\ra G$ as $\lambda\searrow0$. Both sets are therefore dense. The collapse set is dense by \cref{thm:fd-collapse}.

Write $V^+(F):=\limsup_{\gd\ra1}\overline W(\gd;F)$ and $V^-(F):=\liminf_{\gd\ra1}\underline W(\gd;F)$. Since $V^-\leq V^+$, both are at most $\overline W^{CI}$ on $\mathcal F^{\mathrm{collapse}}$ and at least $R$ on $\mathcal F^{\mathrm{canonical}}$. By density of $\mathcal F^{\mathrm{collapse}}$ and $\mathcal F^{\mathrm{canonical}}$, each takes values differing by at least $R-\overline W^{CI}>0$ in every nonempty open set. Hence neither map is continuous at any process.

\subsection{Proof of \cref{thm:generic-instability}}
\label{sec:proof-generic-instability}

For integers $k\geq1$ and $N\geq2$, define $\mathcal U_{k,N}:=\bigcup_{\gd\in(1-1/N,1)}\{F\in \Delta(H):\underline W(\gd;F)>R-1/k\}$ and $\mathcal V_{k,N}:=\bigcup_{\gd\in(1-1/N,1)}\{F\in \Delta(H):\overline W(\gd;F)<\overline W^{CI}+1/k\}$. By \cref{lem:payoff-semicontinuity} in \cref{sec:regularity}, at each fixed $\gd$, $\underline W(\gd;\cdot)$ is lower semicontinuous and $\overline W(\gd;\cdot)$ is upper semicontinuous. Each strict inequality therefore defines an open set of processes, so both unions are open.

For $F\in\mathcal F^{\mathrm{canonical}}$, the first inequality holds at all sufficiently high discount factors; for $F\in\mathcal F^{\mathrm{collapse}}$, the second does. Hence $\mathcal F^{\mathrm{canonical}}\subseteq\mathcal U_{k,N}$ and $\mathcal F^{\mathrm{collapse}}\subseteq\mathcal V_{k,N}$. The mixture construction in the proof of \cref{thm:nowherecontinuity} establishes the density of $\mathcal F^{\mathrm{canonical}}$. Together with \cref{thm:fd-collapse}, this implies that both unions are dense.

By \eqref{eq:generic-instability}, $\mathcal G=\bigcap_{k\geq1}\bigcap_{N\geq2}(\mathcal U_{k,N}\cap\mathcal V_{k,N})$. Its complement is a countable union of closed nowhere-dense sets, so $\mathcal G$ is comeagre. Also, $\Delta(H)$ is compact metrizable in the finite-history topology and hence a Baire space. The Baire category theorem \citep[Chapter 8, Theorem 48.2]{munkres2000topology} implies that $\mathcal G$ is dense.

Finally, suppose that $R>\overline W^{CI}$. Since $\underline W\leq\overline W$, $F\in\mathcal F^{\mathrm{collapse}}$ satisfies $\limsup_{\gd\ra1}\underline W(\gd;F)\leq\overline W^{CI}<R$, whereas $F\in\mathcal F^{\mathrm{survival}}$ satisfies $\liminf_{\gd\ra1}\overline W(\gd;F)>\overline W^{CI}$. Each condition excludes membership in $\mathcal G$. Thus both sets lie in the meagre complement of $\mathcal G$ and are meagre.

\subsection{Proof of \cref{prop:standard-misspecifications2}}

Set $\eta_n:=\sup_{s\in Y^k}\KL{\widehat\rho}{\pi^n(\cdot\mid s)}$, which is finite for all sufficiently large $n$. Fix such an $n$. After the first $k$ signals, the relative-entropy loss of each conditional prediction is at most $\eta_n$. The chain rule therefore gives $\KL{\widehat P_T}{F_T^n}\leq\KL{\widehat P_k}{F_k^n}+(T-k)\eta_n$ for every $T>k$. The first term on the right side of the inequality is finite because $F_k^n$ has full support; when $k=0$, it is zero. Dividing by $T$ and letting $T\ra\infty$ yields $0\leq\overline D(\widehat P\parallel F^n)\leq\eta_n$. Since $\eta_n\ra0$, \cref{prop:entropy-rate-robustness} gives \eqref{eq:entropyrobust}.

\subsection{Proof of \cref{prop:vanish}}

Fix $\gd\in(0,1)$ and $\gs\in\Gs(\gd,F)$. Write $\ga_t(h_t):=\gs_t^1(h_t,\xi^0)$, $\gb_t(h_t):=\gs_t^2(h_t)$, and $\mu_t(h_t):=\mu_t^{\gs^1}(h_t)$. Let $r:=\max_{\ga\in\Delta(A)}\sum_y\sqrt{\rho_\ga(y)f(y)}$. The maximum exists by continuity in $\ga$ and compactness of $\Delta(A)$. By the Cauchy-Schwarz inequality, the sum is at most one, with equality only if $\rho_\ga=f$. Commitment separation excludes this equality, so $r<1$. Define $Z_t(h_t):=F_t(h_t)/L_t^{\gs^1}(h_t)$, with $Z_0=1$. Full-support monitoring makes the denominator positive at every history. Since $F$ is i.i.d.\ with marginal $f$, $Z_{t+1}(h_ty)=Z_t(h_t)f(y)/\rho_{\ga_t(h_t)}(y)$. Taking the conditional expectation under the normal type gives $\bE_\gs[\sqrt{Z_{t+1}}\mid h_t,\xi^0]=\sqrt{Z_t(h_t)}\sum_y\sqrt{\rho_{\ga_t(h_t)}(y)f(y)}\leq r\sqrt{Z_t(h_t)}$. Taking expectations over histories and iterating yields $\bE_\gs[\sqrt{Z_t}\mid\xi^0]\leq r^t$.

Let $\varphi:=\mu_0/(1-\mu_0)$. Bayes' rule and $x/(1+x)\leq\sqrt x$ give $\mu_t=\varphi Z_t/(1+\varphi Z_t)\leq\sqrt{\varphi Z_t}$, hence $\bE_\gs[\mu_t\mid\xi^0]\leq\sqrt\varphi\,r^t$. The constants are independent of the equilibrium and discount factor. Summing with weights $(1-\gd)\gd^t$, and interchanging expectation and summation, gives $0\leq\sup_{\gs\in\Gs(\gd,F)}\bE_\gs[(1-\gd)\sum_{t=0}^\infty\gd^t\mu_t\mid\xi^0]\leq\sqrt\varphi(1-\gd)/(1-\gd r)\ra0$ as $\gd\ra1$. This proves \eqref{eq:vanishingrep}.

Finally, at each history, $q_t^{\gs^1}=\mu_tf+(1-\mu_t)\rho_{\ga_t}$. For any mixed action $\gb$, this gives $\bE_{q_t^{\gs^1}}[\tilde v(\gb,y)]-\bE_{\rho_{\ga_t}}[\tilde v(\gb,y)]=\mu_t\bigl(\bE_f[\tilde v(\gb,y)]-\bE_{\rho_{\ga_t}}[\tilde v(\gb,y)]\bigr)$. Both expectations in parentheses lie in $[-\|\tilde v\|_\infty,\|\tilde v\|_\infty]$, so the absolute difference between the expected payoffs under $q_t^{\gs^1}$ and $\rho_{\ga_t}$ is at most $2\|\tilde v\|_\infty\mu_t$.

Apply this bound to the equilibrium action $\gb_t$ and any alternative $\gb'$. It gives $\bE_{\rho_{\ga_t}}[\tilde v(\gb',y)-\tilde v(\gb_t,y)]\leq\bE_{q_t^{\gs^1}}[\tilde v(\gb',y)-\tilde v(\gb_t,y)]+4\|\tilde v\|_\infty\mu_t$. The expectation on the right is nonpositive because $\gb_t$ is a best response to $q_t^{\gs^1}$. Maximizing over $\gb'$ therefore yields $0\leq\ell_t^\gs\leq4\|\tilde v\|_\infty\mu_t$. Multiplying this inequality by $(1-\gd)\gd^t$, summing over $t$, and taking expectations conditional on the normal type gives 
\begin{align*}
0 \le \bE_\gs\!\left[(1-\gd)\sum_{t=0}^\infty\gd^t\ell_t^\gs\middle|\xi^0 \right] \le 4\|\tilde v\|_\infty\bE_\gs\!\left[(1-\gd)\sum_{t=0}^\infty\gd^t\mu_t \middle| \xi^0 \right] 
\le \frac{4\|\tilde v\|_\infty\sqrt\varphi(1-\gd)}{1-\gd r}.
\end{align*}
The upper bound is independent of $\gs$, so it also bounds the supremum over equilibria. Since $r<1$, it tends to zero as $\gd\ra1$, proving \eqref{eq:avglossboundmain}.

\subsection{Proof of \cref{thm:controlled-fragility}}
\label{sec:proofcontrolledextension}

The proof is analogous to those of \cref{thm:fd-collapse}, \cref{thm:nowherecontinuity}, and \cref{thm:generic-instability}. To avoid repetition, I only describe the necessary modifications. Let $\underline\rho:=\min_{a,b,y}\rho(y\mid a,b)>0$.

\paragraph{Collapse.} In \cref{lem:tciub}, replace the best-response loss by $\ell_2^c(\ga,\gb):=\max_{b\in B}v(\ga,b)-v(\ga,\gb)$ and the score by $S^c(a;\gb,z):=u(a,\gb)-\sum_y\rho_{a,\gb}(y)z(y)$. The maximal-score characterization, compactness argument, and backward recursion are unchanged: the normal type's deviations hold $\gb$ fixed, $\ell_2^c$ is continuous in $(\ga,\gb)$, each score is continuous in $(\gb,z)$, and $\rho_{a,\gb}(y)\geq\underline\rho$. Thus \cref{lem:tciub} holds with the complete-information benchmark of the generalized game.

At a history with posterior probability $\mu$ and normal-type action $\ga$, each action-contingent prediction differs from $\rho_{\ga,b}$ by at most $\mu$ in total variation. The equilibrium response therefore has best-response loss at most $4\|\tilde v\|_\infty\mu$. Moreover, over any $s$ subsequent signals, posterior odds increase by a factor of at most $\underline\rho^{-s}$, uniformly over both players' strategies. Hence the proofs of \cref{lem:small-posterior-ci,prop:genfalsification} apply with likelihoods $F_t^{\gs^2}$ and $L_t^\gs$, provided that the dates $m_{\gd,K}$ are independent of both players' strategies and the probability in \eqref{eq:early-separation-likelihood} tends to one as $\gd \ra 1$ uniformly over all strategy profiles.

For any $G\in\mathcal F_B$ and $N\geq1$, retain its conditional signal distributions before date $N$. Thereafter, use \eqref{eq:full-support-collapse-tail} with $T=N$ and $\widehat\rho$ replaced by any full-support distribution $\nu$, independently of $b$. The likelihood estimates in the proof of \cref{thm:fd-collapse} hold uniformly over both players' strategies because every normal-type signal probability is at least $\underline\rho$. The resulting model induces collapse and satisfies $F_T^{\gs^2}=G_T^{\gs^2}$ for every $T\leq N$ and every $\gs^2$. Letting $N\ra\infty$ proves density. Taking $G=\widehat F$ gives the full-support sequence agreeing exactly with the true commitment process through horizon $N$, under every short-lived-player strategy.

\paragraph{Fragility to subjective inference.} The entropy argument requires the following replacement for \eqref{eq:margin}: for each $w<R^c$, there is $m>0$ such that $\gb\in\mathcal B(q)$ and $u(\widehat\ga,\gb)<w$ imply $D(\widehat\rho_\gb\parallel q^\gb)\geq m$. To see this, suppose there were pairs $(q_j,\gb_j)$ satisfying these conditions with divergence tending to zero. Compactness gives a convergent subsequence. The closed graph of $\mathcal B$ preserves the best-response condition, and \citeauthor{pinsker1964information}'s (\citeyear{pinsker1964information}) inequality and continuity of the mixture distributions give $q^\gb=\widehat\rho_\gb$ at the limit. Continuity of $u(\widehat\ga,\cdot)$ then gives a response $\gb\in\BR^2_0(\widehat\ga)$ with payoff at most $w<R^c$, a contradiction.

In the proof of \cref{thm:nowherecontinuity}, the mixture construction gives nearby processes whose finite-horizon relative entropy is bounded uniformly over horizons. Here the history distributions also depend on the short-lived players' strategies, so I construct nearby models with a bound uniform over those strategies as well. Fix $G\in\mathcal F_B$, $N\geq1$, and $\ve\in(0,1)$. Define $H^{N,\ve}$ by using $(1-\ve)G_t(\cdot\mid h_t,b)+\ve\widehat\rho_b$ before date $N$ and $\widehat\rho_b$ thereafter. Coupling each of the first $T\leq N$ draws with the corresponding draw under $G$, a replacement occurs with probability at most $T\ve$, so $d_T(H^{N,\ve},G)\leq T\ve$. Before date $N$, at every history and action $b$, the subjective probability of each signal $y$ is at least $\ve\widehat\rho_b(y)$. From date $N$ onward, the subjective signal distribution equals the true distribution $\widehat\rho_b$. The chain rule therefore gives $D(\widehat P_T^{\gs^2}\parallel(H^{N,\ve})_T^{\gs^2})\leq N\log(1/\ve)$ for every $T$, uniformly over $\gs^2$.

Fix an equilibrium $\gs$ under $H:=H^{N,\ve}$, and write $\gb_t:=\gs_t^2(h_t)$ and $q_t:=\mathbf q_t^{\gs,H}(h_t)$. Under the commitment-action deviation, histories follow $\widehat P^{\gs^2}$. Let $Q^{\gs,H}:=\mu_0H^{\gs^2}+(1-\mu_0)L^\gs$. Conditional on $h_t$, this law predicts $q_t^{\gb_t}$, while the commitment-action deviation generates $\widehat\rho_{\gb_t}$. The chain-rule argument in \cref{lem:finite-horizon-budget} therefore gives $(1-\gd)\sum_{t\geq0}\gd^t\bE_{\widehat P^{\gs^2}}[D(\widehat\rho_{\gb_t}\parallel q_t^{\gb_t})]\leq(1-\gd)[-\log\mu_0+N\log(1/\ve)]$. For fixed $N$ and $\ve$, this bound vanishes as $\gd\ra1$, uniformly over equilibria. Using the margin $m$ above, the payoff argument for \eqref{eq:finite-patience-robustness} gives $\liminf_{\gd\ra1}\underline W(\gd;H^{N,\ve})\geq R^c$. Choosing $\ve_N\ra0$ therefore proves density of $\mathcal F_B^{\mathrm{canonical}}$. When $R^c>\overline W^{CI}$, the density comparison in the proof of \cref{thm:nowherecontinuity} gives density of survival and nowhere continuity of both patient-payoff maps.

\paragraph{Fragility to discounting.} It remains to verify the Baire property and payoff semicontinuity in the finite-history topology. Let $d(F,G):=\sum_{T\geq1}2^{-T}d_T(F,G)$, which induces the finite-history topology, and identify models at distance zero to obtain a metric space $X_B$. For fixed $T$, the distance $\|F_T^{\gs^2}-G_T^{\gs^2}\|_{\mathrm{TV}}$ depends continuously on the finitely many conditional signal distributions and mixed actions before date $T$. These mixed actions range over a fixed compact product of simplexes, so maximizing over them preserves continuity in $(F,G)$. Thus $d_T(F,G)$ is continuous in $(F,G)$ for the product topology on the families of conditional signal distributions. Since the series defining $d$ converges uniformly, the projection from the compact product $\prod_{t,h_t,b}\Delta(Y)$ onto $X_B$ is continuous. Thus $X_B$ is compact metrizable and Baire. In the finite-history topology, open sets in $\mathcal F_B$ are inverse images of open sets in $X_B$, so $\mathcal F_B$ is Baire as well.

I next show that the short-lived players' predictive signal distribution is jointly continuous in their subjective model and strategy profile, including at subjective-null histories. Fix $(t,h_t,b)$ and let $\gs^{2,h_t,b}$ agree with $\gs^2$ except that it chooses $b$ at $h_t$. Since this change leaves the probability of reaching $h_t$ unchanged, $F_t^{\gs^2}(h_t)F_t(y\mid h_t,b)=F_{t+1}^{\gs^{2,h_t,b}}(h_ty)$. Thus
\begin{align}
q_{t,b}^{\gs,F}(y\mid h_t)
=\frac{\mu_0F_{t+1}^{\gs^{2,h_t,b}}(h_ty)+(1-\mu_0)L_t^\gs(h_t)\rho_{\gs_t^1(h_t,\xi^0),b}(y)}
{\mu_0F_t^{\gs^2}(h_t)+(1-\mu_0)L_t^\gs(h_t)}.
\label{eq:controlled-continuous-prediction}
\end{align}
The two subjective history probabilities vary by at most $d_{t+1}$ and $d_t$ when the model changes, uniformly over strategies; at date zero, the history probability is one. For a fixed model, all terms are continuous in the finitely many strategy coordinates involved. Since the denominator is at least $(1-\mu_0)\underline\rho^t>0$, \eqref{eq:controlled-continuous-prediction} implies the desired joint continuity.

The proof of \cref{lem:payoff-semicontinuity} now applies in the finite-history topology. Using the dense collapse and canonical sets established above, the Baire argument in \cref{thm:generic-instability} applies with $R^c$ in place of $R$. Thus $\mathcal G_B$ is comeagre, and when $R^c>\overline W^{CI}$ the same disjointness argument makes the collapse and survival sets meagre.

\end{appendices}

\pagebreak

\addtocontents{toc}{\protect\setcounter{tocdepth}{1}}
\renewcommand\bibname{References}

{
\bibliographystyle{chicago}
\bibliography{references}
}

\pagebreak

\setcounter{section}{0}
\renewcommand{\thesection}{\Roman{section}}
\renewcommand{\theHsection}{supplementary.\Roman{section}}

\pagebreak
\clearpage
\pagenumbering{arabic}
\renewcommand{\thepage}{SA-\arabic{page}}
\section{Supplementary Appendix}

\subsection{Vanishing average value of commitment}
\label[onlineappendix]{sec:vanishingaverage}

Collapse does not require the hypotheses in \cref{prop:genfalsification}. This section gives an alternative sufficient condition for collapse, which can hold even when the commitment type remains statistically plausible, because the normal type's payoff advantage from the short-lived players' responses vanishes in discounted average. 

Recall the score $S(a;\gb,z):=u(a,\gb)-\sum_{y\in Y}\rho(y\mid a)z(y)$ for $a\in A$, $\gb\in\Delta(B)$, and $z\in\bR_+^Y$ from the proof of \cref{prop:genfalsification}. For $\ve\geq0$, call $q\in\Delta(Y)$ an \emph{$\ve$-no-value prediction} if the following holds. Whenever $\ga\in\Delta(A)$, $\mu\in[0,1)$, $\gb\in\BR^2(\mu q+(1-\mu)\rho_\ga)$, and $z\in\bR_+^Y$ satisfy $\operatorname{supp}(\ga)\subseteq\arg\max_{a\in A}S(a;\gb,z)$, there exists $\gb^0\in\BR^2(\rho_\ga)$ such that $\operatorname{supp}(\ga)\subseteq\arg\max_{a\in A}S(a;\gb^0,z)$ and $\max_{a\in A}S(a;\gb,z)\leq\max_{a\in A}S(a;\gb^0,z)+\ve$. Let $\mathcal Q^\ve$ denote the set of these predictions. The condition allows the short-lived player's response to be replaced by a complete-information best response that keeps $\ga$ optimal for the same $z$, without lowering the normal type's maximal score by more than $\ve$.

\begin{prop*}
\label{prop:genfalsification2}
Fix $F\in\Delta(H)$. Suppose there are an integer $\tau\geq1$ and a bounded sequence $(\ve_t)_{t=\tau}^\infty$ of nonnegative numbers such that $F(\cdot\mid h_t)\in\mathcal Q^{\ve_t}$ for every $t\geq\tau$ and every history $h_t$ with $F_t(h_t)>0$, and
\begin{align}
\lim_{\gd\ra1}
(1-\gd)
\sum_{t=\tau}^{\infty}
\gd^{t-\tau}\ve_t
=0.
\label{eq:vanishing-average-commitment-value}
\end{align}
Then $F\in\mathcal F^{\mathrm{collapse}}$. If $\mathcal Q^0\neq\varnothing$, there exists a sequence $(F^n)_{n = 1}^\infty$ along which misspecification vanishes empirically and each $F^n$ satisfies these hypotheses with $\ve_t^n=0$ for every $t\geq\tau_n$.
\end{prop*}

For a concrete example, take the product-choice game in \cref{sec:productchoice} with $p=3/4$ and $p'=1/4$, and let the consumers' subjective process $F$ be i.i.d.\ with marginal $f$ such that $f(y_h)=1/4$ and $f(y_l)=3/4$. Since $\rho_\ga(y_h)\geq1/4$ for every $\ga$, mixing $\rho_\ga$ with $f$ weakly lowers the predicted probability of $y_h$. Consumers prefer $b_h$ above probability $1/2$ and $b_l$ below it, and are indifferent at equality. Thus, for every $\mu\in[0,1)$, any best response $\gb$ to $\mu f+(1-\mu)\rho_\ga$ can be replaced by a best response $\gb^0$ to $\rho_\ga$ with $\gb^0(b_h)\geq\gb(b_h)$. For either firm action $a$, $S(a;\gb^0,z)-S(a;\gb,z)=2[\gb^0(b_h)-\gb(b_h)]\geq0$. The replacement adds the same nonnegative amount to both action scores, preserving the maximizing actions for the same $z$. Hence $f\in\mathcal Q^0$. The process $F$ therefore satisfies \cref{prop:genfalsification2} with $\tau=1$ and $\ve_t=0$, implying $F\in\mathcal F^{\mathrm{collapse}}$. The hypotheses in \cref{prop:genfalsification} nevertheless fail: persistent low effort generates exactly $F$, so the likelihood ratio in \eqref{eq:early-separation-likelihood} equals one at every history under this normal-type strategy.

In this product-choice example, the conclusions of \cref{thm:fd-collapse} continue to hold under a uniform positive lower bound on subjective conditional signal probabilities. Fix $0<\ve\leq1/4$, and let $\mathcal C_\ve$ be the class of full-support subjective processes satisfying $G(y\mid h_t)\geq\ve$ for every date, history, and signal. Given any $G\in\mathcal C_\ve$, let $F^n$ agree with $G$ through horizon $n$ and use $f$ after every history at dates $t\geq n$. Then $F^n\in\mathcal C_\ve$, and \cref{prop:genfalsification2} applies with $\tau=n$ and $\ve_t=0$ thereafter, giving $F^n\in\mathcal F^{\mathrm{collapse}}$. For every fixed horizon $T$, $F_T^n=G_T$ whenever $n\geq T$. Thus $F^n\ra G$, proving that collapse is dense in $\mathcal C_\ve$ under the relative finite-history topology. Taking $G=\widehat P$ gives an empirically vanishing sequence of collapse processes with $F_T^n=\widehat P_T$ whenever $T\leq n$, while every subjective conditional signal probability remains at least $\ve$.

\begin{proof}[Proof of \cref{prop:genfalsification2}]
Fix $F$, $\tau$, and $(\ve_t)_{t\geq\tau}$ as in the proposition, and fix $\gd\in(0,1)$ and an equilibrium strategy profile $\gs\in\Gs(\gd,F)$. Recall that the complete-information maximal score $\kappa_0$ in \eqref{eq:claim1-kappa} equals $\overline W^{CI}$. Let $V_t^\gs(h_t)$ denote the normal type's normalized continuation payoff at $h_t$, with discounting measured from date $t$, and set $M_t^\gs:=\sup_{h_t\in Y^t}V_t^\gs(h_t)$. These quantities are bounded by the minimum and maximum stage payoffs.

Fix $t\geq\tau$ and $h_t$, and write $\ga:=\gs_t^1(h_t,\xi^0)$ and $\gb:=\gs_t^2(h_t)$. Full-support monitoring makes every history possible under the normal type. Equilibrium optimality therefore gives $V_t^\gs(h_t)\geq(1-\gd)u(a,\gb)+\gd\sum_y\rho(y\mid a)V_{t+1}^\gs(h_ty)$ for every $a\in A$, with equality whenever $a\in\operatorname{supp}(\ga)$. Set $z(y):=\gd[M_{t+1}^\gs-V_{t+1}^\gs(h_ty)]/(1-\gd)\geq0$ and $r_t^\gs(h_t):=[V_t^\gs(h_t)-\gd M_{t+1}^\gs]/(1-\gd)$. Subtracting $\gd M_{t+1}^\gs$ from the deviation inequalities and dividing by $1-\gd$ gives $r_t^\gs(h_t)=\max_{a\in A}S(a;\gb,z)$ and
\begin{align}
\operatorname{supp}(\ga)\subseteq\arg\max_{a\in A}S(a;\gb,z).
\label{eq:part-three-local-support}
\end{align}
If $F_t(h_t)>0$, then $\gb\in\BR^2(\mu_t^{\gs^1}(h_t)F(\cdot\mid h_t)+(1-\mu_t^{\gs^1}(h_t))\rho_\ga)$, with $\mu_t^{\gs^1}(h_t)<1$ by full-support monitoring. Since $F(\cdot\mid h_t)\in\mathcal Q^{\ve_t}$, there is a complete-information best response $\gb^0\in\BR^2(\rho_\ga)$ that preserves \eqref{eq:part-three-local-support} for the same $z$ and satisfies $r_t^\gs(h_t)\leq\max_{a\in A}S(a;\gb^0,z)+\ve_t$. The tuple $(\ga,\gb^0,z)$ is feasible in \eqref{eq:claim1-kappa}, so $r_t^\gs(h_t)\leq\kappa_0+\ve_t$. If $F_t(h_t)=0$, the posterior is zero and $\gb\in\BR^2(\rho_\ga)$; the original tuple $(\ga,\gb,z)$ is then feasible, giving $r_t^\gs(h_t)\leq\kappa_0$. Substituting the definition of $r_t^\gs(h_t)$ and taking the supremum over histories yields $M_t^\gs\leq(1-\gd)(\kappa_0+\ve_t)+\gd M_{t+1}^\gs$ for every $t\geq\tau$.

Iterating this inequality $N$ times from date $\tau$ gives $M_\tau^\gs\leq(1-\gd^N)\kappa_0+(1-\gd)\sum_{t=\tau}^{\tau+N-1}\gd^{t-\tau}\ve_t+\gd^N M_{\tau+N}^\gs$. Since continuation payoffs and $(\ve_t)_{t=\tau}^\infty$ are bounded, letting $N\to\infty$ gives
\begin{align}
M_\tau^\gs\leq\kappa_0+(1-\gd)\sum_{t=\tau}^{\infty}\gd^{t-\tau}\ve_t.
\label{eq:approximate-no-value-tail-bound}
\end{align}
Let $\overline u:=\max_{a,b}u(a,b)$. Bounding payoffs before date $\tau$ by $\overline u$ and continuation payoffs at that date by $M_\tau^\gs$ gives $U(\gs;\gd,F)\leq(1-\gd^\tau)\overline u+\gd^\tau M_\tau^\gs$. Substituting \eqref{eq:approximate-no-value-tail-bound} yields a bound independent of the equilibrium strategy profile. Taking the supremum over equilibria and then the limit superior as $\gd\ra1$, we obtain $\limsup_{\gd\ra1}\overline W(\gd;F)\leq\kappa_0=\overline W^{CI}$: the initial-payoff term $(1-\gd^\tau)\overline u$ vanishes because $\tau$ is fixed, $\gd^\tau\kappa_0\ra\kappa_0$, and the error term $\gd^\tau(1-\gd)\sum_{t=\tau}^{\infty}\gd^{t-\tau}\ve_t$ vanishes by \eqref{eq:vanishing-average-commitment-value}.

Finally, suppose $\mathcal Q^0\neq\varnothing$ and choose $f^0\in\mathcal Q^0$. For each $n\geq1$, let $F^n$ have independent signals with marginal $\widehat\rho$ at times $t<n$ and marginal $f^0$ thereafter. Each $F^n$ satisfies the hypotheses with $\tau_n=n$ and $\ve_t^n=0$ for every $t\geq n$, so the preceding argument gives $F^n\in\mathcal F^{\mathrm{collapse}}$. Moreover, $F_T^n=\widehat P_T$ whenever $n\geq T$, proving that misspecification vanishes empirically.
\end{proof}

\subsection{Joint-limit robustness}
\label[onlineappendix]{sec:jointlimit}

In this section I provide a joint-limit version of \cref{prop:entropy-rate-robustness}.

\begin{prop*}
\label{prop:jointlimit}
Let $(\gd_n,F^n)_{n=1}^{\infty}$ satisfy $\gd_n\ra1$ and $D_{\gd_n}(\widehat P\parallel F^n)\ra0$. Then $\liminf_{n\to\infty}\underline W(\gd_n;F^n)\geq R$.
\end{prop*}

\begin{proof}[Proof of \cref{prop:jointlimit}]
Fix $w<R$. The coefficient in \eqref{eq:finite-patience-robustness} is finite and independent of $n$, while $-(1-\gd_n)\log\mu_0+D_{\gd_n}(\widehat P\parallel F^n)\ra0$. Applying this bound to each $(\gd_n,F^n)$ and taking the limit inferior gives $\liminf_{n\to\infty}\underline W(\gd_n;F^n)\geq w$. Letting $w\nearrow R$ proves the claim.
\end{proof}

Unlike \cref{prop:entropy-rate-robustness}, \cref{prop:jointlimit} concerns vanishing $D_{\gd_n}(\widehat P\parallel F^n)$ but not vanishing entropy-rate loss $\overline D(\widehat P\parallel F^n)$. The entropy rate ignores any finite initial loss, but this loss can grow along the sequence and remain significant at the corresponding discount factors. To see this, let $F^n$ have independent signals, with marginal $f^n$ at date $0$ and $\widehat\rho$ thereafter, where $f^n$ has full support and satisfies $\KL{\widehat\rho}{f^n}=n$. Such a marginal exists: starting from $\widehat\rho$, continuously reduce one signal's probability toward zero and redistribute the removed probability among the other signals. The divergence varies continuously from zero to infinity, so it equals $n$ at some full-support marginal. The chain rule of relative entropy gives $\KL{\widehat P_T}{F_T^n}=n$ for every $T\geq1$. Thus $\overline D(\widehat P\parallel F^n)=0$ for every $n$, whereas $D_\gd(\widehat P\parallel F^n)=(1-\gd)n$. Taking $\gd_n=1-1/(n+1)$ yields $D_{\gd_n}(\widehat P\parallel F^n)=n/(n+1)\ra1$.

\subsection{Payoff guarantee under a fixed misspecification}
\label[onlineappendix]{sec:genrepbuild}

This section derives a payoff guarantee for a fixed subjective process $F$. The bound optimizes over deviations, balancing the payoffs they support against their
likelihood loss under $F$.

Let $Z$ be the set of normal-type behavior strategies, and let $L^\zeta$ denote the signal process generated by $\zeta\in Z$. For $P,F\in\Delta(H)$, define $D^*(P\parallel F):=\limsup_{\gd\ra1}D_\gd(P\parallel F)$, where $D_\gd$ is defined in \eqref{eq:general-discounted-entropy}. Throughout this section, use the convention $0\cdot\infty:=0$.

For $\lambda\geq0$, define the entropy-compensated guarantee from $\ga\in\Delta(A)$ by $\gam_\lambda(\ga):=\inf_{q\in\Delta(Y)}[\min_{\gb\in\BR^2(q)}u(\ga,\gb)+\lambda D(\rho_\ga\parallel q)]$. The scalar $\lambda$ converts entropy loss into payoff units. For any normal-type deviation $\zeta$, let $\Gamma_\lambda(\zeta):=\liminf_{\gd\ra1}\bE_{L^\zeta}[(1-\gd)\sum_{t=0}^{\infty}\gd^t\gam_\lambda(\zeta_t(h_t))]$ denote the patient lower limit of its discounted-average entropy-compensated guarantee.

\begin{prop*}
\label{prop:genbound}
For every subjective process $F$,
\begin{align}
\liminf_{\gd\ra1}
\underline W(\gd;F)
\geq
\sup_{\zeta\in Z,\lambda\geq0}
\left\{
\Gamma_\lambda(\zeta)
-
\lambda D^*(L^\zeta\parallel F)
\right\}.
\label{eq:genbound}
\end{align}
\end{prop*}

Under correct specification, this bound recovers the canonical guarantee. Let $\widehat\zeta$ prescribe $\widehat\ga$ after every history. Then $L^{\widehat\zeta}=F=\widehat P$, so $D^*(L^{\widehat\zeta}\parallel F)=0$ and $\Gamma_\lambda(\widehat\zeta)=\gam_\lambda(\widehat\ga)$. Fix $w<R$. If $w\leq\underline u_{\widehat\ga}$, then $\gam_0(\widehat\ga)\geq w$. Otherwise, choose $\lambda>0$ with $\lambda m_{\widehat\ga}(w)\geq w-\underline u_{\widehat\ga}$, using the positive margin from \eqref{eq:margin}; any positive $\lambda$ suffices if the margin is infinite. For any prediction $q$, either $\phi_{\widehat\ga}(q)\geq w$, or $\phi_{\widehat\ga}(q)\geq\underline u_{\widehat\ga}$ and $D(\widehat\rho\parallel q)\geq m_{\widehat\ga}(w)$. Thus $\phi_{\widehat\ga}(q)+\lambda D(\widehat\rho\parallel q)\geq w$ in both cases, giving $\gam_\lambda(\widehat\ga)\geq w$. Applying \eqref{eq:genbound} and letting $w\nearrow R$ proves \eqref{eq:FLbound}.

\begin{eg*}[Noncanonical reputation guarantee]
\label{eg:noncanonical-deviation}\normalfont
Consider the product-choice game in \cref{sec:productchoice} with $p=3/4$ and $p'=1/4$. Consumers choose $b_h$ when their predictive probability of $y_h$ exceeds $1/2$, choose $b_l$ when it is below $1/2$, and may mix at equality. Suppose $F$ is i.i.d.\ with $F(y_h)=5/8$. For $s,r\in(0,1)$, write $\diff(s,r):=s\log(s/r)+(1-s)\log[(1-s)/(1-r)]$.

First consider the normal type's deviation to persistent high effort, denoted by $\widehat\zeta$. It gives payoff $2$ against $b_h$ and zero against $b_l$. The short-lived players' predictive probabilities of signal $y_h$ weakly below $1/2$ permit a best response $b_l$ and have minimum entropy loss $m_c:=\diff(3/4,1/2)$. A predictive probability of signal $y_h$ above $1/2$ induces best response $b_h$, and the prediction $3/4$ has zero entropy loss. Hence $\Gamma_\lambda(\widehat\zeta)=\min\{2,\lambda m_c\}$. Moreover, $D^*(L^{\widehat\zeta}\parallel F)=e_c:=\diff(3/4,5/8)$, with $0<e_c<m_c$. Optimizing $\min\{2,\lambda m_c\}-\lambda e_c$ over $\lambda\geq0$ gives $2-2e_c/m_c<2$, attained at $\lambda=2/m_c$.

Now let $\zeta_{nc}$ prescribe $\ga_{nc}:=(3/4) \circ \{a_h\}+(1/4) \circ \{a_l\}$ after every history. This deviation generates $y_h$ with probability $5/8$, exactly matching $F$, and gives payoff $9/4$ against $b_h$ and $1/4$ against $b_l$. With $m_{nc}:=\diff(5/8,1/2)$, the same calculation gives $\Gamma_\lambda(\zeta_{nc})=\min\{9/4,1/4+\lambda m_{nc}\}$ and $D^*(L^{\zeta_{nc}}\parallel F)=0$. Taking $\lambda=2/m_{nc}$ in \eqref{eq:genbound} yields $\liminf_{\gd\ra1}\underline W(\gd;F)\geq9/4>2$. Occasional low effort therefore gives a stronger payoff guarantee than persistent high effort: it matches the subjective signal process while still inducing consumers to choose $b_h$.
\end{eg*}

\begin{proof}[Proof of \cref{prop:genbound}]
Fix $\zeta\in Z$ and $\lambda\geq0$, and write $P:=L^\zeta$. Set $K_\gd(\zeta,\lambda):=\bE_P[(1-\gd)\sum_{t=0}^{\infty}\gd^t\gam_\lambda(\zeta_t(h_t))]$ and $M_\gd(\zeta,F):=-(1-\gd)\log\mu_0+D_\gd(P\parallel F)$. Nonnegativity of relative entropy and the choice $q=\rho_\ga$ imply $\underline u\leq\gam_\lambda(\ga)\leq\overline u$, where $\underline u=\min_{a,b}u(a,b)$ and $\overline u=\max_{a,b}u(a,b)$. Thus $K_\gd$ is bounded, and its limit inferior as $\gd\ra1$ is $\Gamma_\lambda(\zeta)$.

Consider any strategy profile $\gs$ satisfying the short-lived players' best-response conditions. Their subjective distribution over length-$T$ histories satisfies $Q_T^{\gs^1}=\mu_0F_T+(1-\mu_0)L_T^{\gs^1}\geq\mu_0F_T$. Consequently, $\KL{P_T}{Q_T^{\gs^1}}\leq-\log\mu_0+\KL{P_T}{F_T}$. Multiplying by $(1-\gd)^2\gd^{T-1}$, summing over $T\geq1$, and applying \eqref{eq:discounted-chain-rule} gives
\begin{align}
\bE_P\!\left[(1-\gd)\sum_{t=0}^{\infty}\gd^t
\KL{\rho_{\zeta_t(h_t)}}{q_t^{\gs^1}(\cdot\mid h_t)}\right]
\leq M_\gd(\zeta,F).
\label{eq:gen-entropy-budget}
\end{align}
Here histories follow the deviation distribution $P$, while the short-lived players retain their strategies and predictions under $\gs$. Since $\gs_t^2(h_t)\in\BR^2(q_t^{\gs^1}(\cdot\mid h_t))$, the definition of $\gam_\lambda$ gives $u(\zeta_t(h_t),\gs_t^2(h_t))\geq\gam_\lambda(\zeta_t(h_t))-\lambda\KL{\rho_{\zeta_t(h_t)}}{q_t^{\gs^1}(\cdot\mid h_t)}$. For $\lambda>0$ and finite $M_\gd$, \eqref{eq:gen-entropy-budget} implies
\begin{align}
U((\zeta,\gs^2);\gd,F)
\geq K_\gd(\zeta,\lambda)-\lambda M_\gd(\zeta,F).
\label{eq:gen-deviation-payoff-bound}
\end{align}
If $M_\gd=\infty$, this bound is immediate. If $\lambda=0$, omit the entropy term before taking expectations to obtain $U((\zeta,\gs^2);\gd,F)\geq K_\gd(\zeta,0)$.

Now let $\gs$ be an equilibrium. The normal type's payoff is at least his deviation payoff. Since the preceding bounds are independent of $\gs$, taking the infimum over equilibria gives
\begin{align}
\underline W(\gd;F)
&\geq K_\gd(\zeta,\lambda)-\lambda M_\gd(\zeta,F),
&&\lambda>0,
\label{eq:gen-finite-payoff-bound}\\
\underline W(\gd;F)&\geq K_\gd(\zeta,0).
\label{eq:gen-zero-price-bound}
\end{align}
The prior term vanishes as $\gd\ra1$, so $\limsup_{\gd\ra1}M_\gd(\zeta,F)=D^*(P\parallel F)$. For $\lambda>0$ and finite $D^*(P\parallel F)$, taking the limit inferior as $\gd \ra 1$ in \eqref{eq:gen-finite-payoff-bound}, with $(\zeta,\lambda)$ fixed, gives $\liminf_{\gd\ra1}\underline W(\gd;F)\geq\Gamma_\lambda(\zeta)-\lambda D^*(P\parallel F)$. This inequality is immediate when $D^*(P\parallel F)=\infty$. For $\lambda=0$, taking the limit inferior as $\gd \ra 1$ in \eqref{eq:gen-zero-price-bound} gives $\Gamma_0(\zeta)$ without any restriction on $D^*$. Taking the supremum over fixed pairs $(\zeta,\lambda)$ proves \eqref{eq:genbound}.
\end{proof}

\subsection{Finite-horizon entropy-rate criterion}
\label[onlineappendix]{sec:finitehorizon}

Consider the $T$-period truncation of the model, with the normal type evaluated by his undiscounted average payoff over dates $0,\ldots,T-1$. Let $\Gs_T(F)$ denote the set of subjective Nash equilibria when the short-lived players use $F_T$. For a strategy profile $\gs$, define $U_T(\gs;F):=\bE_\gs[T^{-1}\sum_{t=0}^{T-1}u(a_t,b_t)\mid\xi^0]$ and $\underline W_T(F):=\inf_{\gs\in\Gs_T(F)}U_T(\gs;F)$.

\begin{prop*}
\label{prop:finite-horizon-entropy}
Let $(F^n)_{n=1}^{\infty}$ be a sequence in $\Delta(H)$. If $\overline D(\widehat P\parallel F^n)\ra0$, then $\liminf_{n\to\infty}\liminf_{T\to\infty}\underline W_T(F^n)\geq R$.
\end{prop*}

\begin{proof}[Proof of \cref{prop:finite-horizon-entropy}]
Use the same deviation and payoff-margin argument as in the proof of \cref{prop:entropy-rate-robustness}, replacing discounted averages by averages over dates $0,\ldots,T-1$. The finite-horizon entropy budget in \cref{lem:finite-horizon-budget}, divided by $T$, replaces the discounted entropy budget \eqref{eq:discounted-entropy-budget}. The lemma applies after extending the finite-horizon strategy profile arbitrarily beyond date $T-1$, which leaves all earlier predictions unchanged. Thus, for every $F\in\Delta(H)$, $T\geq1$, and $w<R$, the same argument gives
\begin{align}
\underline W_T(F)
\geq
w-
\frac{\max\{0,w-\underline u_{\widehat\ga}\}}{m_{\widehat\ga}(w)}
\frac{-\log\mu_0+\KL{\widehat P_T}{F_T}}{T}.
\label{eq:finite-horizon-bound}
\end{align}
As in that proof, the penalty---the second term in the lower bound in \eqref{eq:finite-horizon-bound}---is zero when $w\leq\underline u_{\widehat\ga}$ or $m_{\widehat\ga}(w)=\infty$.

For all sufficiently large $n$, $\overline D(\widehat P\parallel F^n)$ is finite. Apply \eqref{eq:finite-horizon-bound} to each such $F^n$ and take the limit inferior as $T\to\infty$. The prior term divided by $T$ vanishes, while the limit superior as $T \ra \infty$ of $\KL{\widehat P_T}{F_T^n}/T$ equals $\overline D(\widehat P\parallel F^n)$. Hence $\liminf_{T\to\infty}\underline W_T(F^n)\geq w-[\max\{0,w-\underline u_{\widehat\ga}\}/m_{\widehat\ga}(w)]\overline D(\widehat P\parallel F^n)$. The coefficient is finite and independent of $n$. Letting $n\to\infty$ and then $w\nearrow R$ proves the claim.
\end{proof}

\subsection{Finite ambiguity about the subjective process}
\label[onlineappendix]{sec:robust}

This section extends the three main results to a normal type who knows only that the short-lived players' subjective process belongs to a nonempty finite set $\mathcal F\subseteq\Delta(H)$. He chooses one strategy and evaluates its worst-case ex ante payoff over $\mathcal F$.

A robust strategy profile is $\gs:=(\gs^1,(\gs^{2,F})_{F\in\mathcal F})$.  The normal type uses the same strategy $\gs^1$ for every candidate, while the short-lived players' strategy $\gs^{2,F}$ may depend on $F$. In this extension, as in the main analysis, strategies are based on public histories.\footnote{Under ambiguity, this restriction is not as innocuous as in my main analysis because the short-lived players' actions may reveal their subjective process. The proof relies on each fixed normal-type strategy generating the same signal process under every candidate. If the normal type conditions his play on observed short-lived-player actions, different candidates may induce different future actions by the normal type and hence different signal processes.} Let $\mu_t^{\gs^1,F}$ and $q_t^{\gs^1,F}$ denote the posterior and prediction associated with $(\gs^1,F)$. A robust equilibrium given $(\gd,\mathcal F)$ is a robust strategy profile satisfying $\gs^1\in\arg\max_{\widetilde\gs^1}\min_{F\in\mathcal F}U((\widetilde\gs^1,\gs^{2,F});\gd,F)$ and $\gs_t^{2,F}(h_t)\in\BR^2(q_t^{\gs^1,F}(\cdot\mid h_t))$ for every $F\in\mathcal F$, $t\geq0$, and $h_t\in Y^t$. Write $\Gs(\gd,\mathcal F)$ for the set of robust equilibria. This set is nonempty for every $\gd\in(0,1)$.\footnote{In a finite truncation, let the normal type randomize over pure strategies. His worst-case payoff is the minimum of finitely many affine functions of his mixed strategy, so it is continuous and concave. Each short-lived player's payoff is continuous in the strategy profile and linear in her own mixed action; full-support monitoring ensures that her posterior is continuous in the normal type's strategy. The best-response correspondences are therefore nonempty, convex-valued, and have closed graphs, so Kakutani's theorem gives an equilibrium. Replace the normal type's mixed strategy by a public behavior strategy that assigns the same conditional probabilities to his actions after every public history. Since signals depend only on his actions, the same replacement works for every candidate and preserves the distributions of public histories, the short-lived players' predictions, and expected payoffs. Extend these strategy profiles arbitrarily beyond their finite horizons. Compactness of the public strategy space gives a subsequence converging at every history as the horizon grows. For fixed $\gd<1$, discounted tails vanish uniformly over strategy profiles. Since $\mathcal F$ is finite, taking limits in the normal type's optimality inequality against each fixed deviation preserves worst-case optimality. Continuity of predictions at each fixed history and the closed graph of $\BR^2$ preserve the short-lived players' best-response conditions.} Let $J_\gd(\gs;\mathcal F):=\min_{F\in\mathcal F}U((\gs^1,\gs^{2,F});\gd,F)$ denote the normal type's worst-case payoff. Define $\underline W(\gd;\mathcal F):=\inf_{\gs\in\Gs(\gd,\mathcal F)}J_\gd(\gs;\mathcal F)$ and $\overline W(\gd;\mathcal F):=\sup_{\gs\in\Gs(\gd,\mathcal F)}J_\gd(\gs;\mathcal F)$.

Fix the number of subjective-process candidates $m\geq1$, and let $\mathfrak F_m:=\{\mathcal F\subseteq\Delta(H):|\mathcal F|=m\}$. Metrize finite-history proximity by $d(F,G):=\sum_{T=1}^{\infty}2^{-T}\|F_T-G_T\|_{\mathrm{TV}}$, and endow $\mathfrak F_m$ with the Hausdorff metric $d_H(\mathcal F,\mathcal E):=\max\{\max_{F\in\mathcal F}\min_{E\in\mathcal E}d(F,E),\allowbreak\max_{E\in\mathcal E}\min_{F\in\mathcal F}d(F,E)\}$. Two ambiguity sets are therefore close when every candidate in either set is close to some candidate in the other.

As in the main text, define $\mathfrak F_m^{\mathrm{collapse}}:=\{\mathcal F\in\mathfrak F_m:\limsup_{\gd\ra1}\overline W(\gd;\mathcal F)\leq\overline W^{CI}\}$, $\mathfrak F_m^{\mathrm{canonical}}:=\{\mathcal F\in\mathfrak F_m:\liminf_{\gd\ra1}\underline W(\gd;\mathcal F)\geq R\}$, and $\mathfrak F_m^{\mathrm{survival}}:=\{\mathcal F\in\mathfrak F_m:\liminf_{\gd\ra1}\underline W(\gd;\mathcal F)>\overline W^{CI}\}$.

\begin{thm*}
\label{thm:finite-ambiguity}
Fix $m\geq1$.
\begin{enumerate}\itemsep0em
\item Collapse: The set $\mathfrak F_m^{\mathrm{collapse}}$ is dense in $\mathfrak F_m$. In particular, there exists a sequence $(\mathcal F^n)_{n=1}^{\infty}$ of ambiguity sets whose candidates all have full support such that for every $n$, $\limsup_{\gd\ra1}\overline W(\gd;\mathcal F^n)\leq\overline W^{CI}$. Moreover, the sequence can be chosen so that $F_T=\widehat P_T$ for every $n$, every $F\in\mathcal F^n$, and every horizon $T=1,\ldots,n$.

\item Fragility to subjective inference: Suppose that $R>\overline W^{CI}$. The set $\mathfrak F_m^{\mathrm{canonical}}$, and therefore $\mathfrak F_m^{\mathrm{survival}}$, is dense in $\mathfrak F_m$. Together with the density of $\mathfrak F_m^{\mathrm{collapse}}$, this implies that the maps $\mathcal F\mapsto\limsup_{\gd\ra1}\overline W(\gd;\mathcal F)$ and $\mathcal F\mapsto\liminf_{\gd\ra1}\underline W(\gd;\mathcal F)$ are nowhere continuous.

\item Fragility to discounting: The set
\begin{align}
\mathcal G_m
:=\left\{\mathcal F\in\mathfrak F_m:
\liminf_{\gd\ra1}\overline W(\gd;\mathcal F)\leq\overline W^{CI}
\text{ and }
\limsup_{\gd\ra1}\underline W(\gd;\mathcal F)\geq R\right\}
\label{eq:finite-ambiguity-genericity}
\end{align}
is comeagre in $\mathfrak F_m$. Therefore, if $R>\overline W^{CI}$, then  $\mathfrak F_m^{\mathrm{collapse}}$ and $\mathfrak F_m^{\mathrm{survival}}$ are meagre.
\end{enumerate}
\end{thm*}



\begin{proof}[Proof of \cref{thm:finite-ambiguity}]
The proof is analogous to those of \cref{thm:fd-collapse}, \cref{thm:nowherecontinuity}, and \cref{thm:generic-instability}. To avoid repetition, I only describe the necessary modifications.

\paragraph{Collapse.} To adapt \cref{prop:genfalsification}, I first recover the continuation optimality needed for \cref{lem:tciub}. Worst-case optimality need not imply optimality against any individual $\gb^i$. To apply \cref{lem:tciub}, I construct an average of these response strategies against which $\gs^1$ is optimal after every public history. Fix $\mathcal F=\{F_1,\ldots,F_m\}$, $\gd\in(0,1)$, and a robust equilibrium strategy profile $\gs$. Write $\gb^i:=\gs^{2,F_i}$, $U_i(\zeta):=U((\zeta,\gb^i);\gd,F_i)$, and $w:=J_\gd(\gs;\mathcal F)$. Let $Z$ denote the normal-type behavior strategies with the product topology. This space is compact, and each $U_i$ is continuous in $\zeta$ by the discounted-tail argument in \cref{lem:payoff-semicontinuity}. Thus the attainable payoff vectors form a compact set $\mathcal V:=\{(U_1(\zeta),\ldots,U_m(\zeta)):\zeta\in Z\}$.

The set $\mathcal V$ is convex: To prove this, fix $\zeta,\zeta'\in Z$ and $\lambda\in[0,1]$. Define a public behavior strategy by $\zeta_t^\lambda(h_t):=[\lambda L_t^\zeta(h_t)\zeta_t(h_t)+(1-\lambda)L_t^{\zeta'}(h_t)\zeta'_t(h_t)]/[\lambda L_t^\zeta(h_t)+(1-\lambda)L_t^{\zeta'}(h_t)]$. The denominator is positive by full-support monitoring. This strategy reproduces the distribution of histories and current actions obtained by choosing $\zeta$ with probability $\lambda$ and $\zeta'$ otherwise at the outset.  Hence $U_i(\zeta^\lambda)=\lambda U_i(\zeta)+(1-\lambda)U_i(\zeta')$ for every $i$, proving that $\mathcal V$ is convex. Applying the minimax theorem to the compact convex set $\mathcal V$ gives $p\in\Delta(\{1,\ldots,m\})$ such that
\begin{align}
\sum_{i=1}^m p_iU_i(\zeta)
\leq w
=\sum_{i=1}^m p_iU_i(\gs^1)
\qquad\text{for every }\zeta\in Z.
\label{eq:robust-supporting-probability}
\end{align}
Set $\overline\gb_t(h_t):=\sum_i p_i\gb_t^i(h_t)$. Every normal-type strategy induces the same signal distribution against all candidates, so its payoff against $\overline\gb$ equals its $p$-weighted payoff across candidates. Thus \eqref{eq:robust-supporting-probability} makes $\gs^1$ a best response to $\overline\gb$, with payoff $w$. In the collapse argument below, continuation payoffs are evaluated against $\overline\gb$.

With this averaged response, I now extend \cref{prop:genfalsification} to finitely many candidates. For $\mathcal F=\{F_1,\ldots,F_m\}$, suppose that, for every $K>0$, common deterministic dates $m_{\gd,K}$ satisfy \eqref{eq:early-separation-date} and $\inf_{\zeta\in Z}L^\zeta(\min_{1\leq i\leq m}\log[L_{m_{\gd,K}}^\zeta(h_{m_{\gd,K}})/(F_i)_{m_{\gd,K}}(h_{m_{\gd,K}})]\geq K/(1-\gd))\ra1$, with the same convention for zero denominators. Thus the likelihood bound holds simultaneously for all candidates, with probability tending to one uniformly over normal-type strategies. I show that $\mathcal F\in\mathfrak F_m^{\mathrm{collapse}}$.

The proof of \cref{lem:small-posterior-ci} requires one modification. Fix $\eta>0$ and choose $\ve,L,K,\bar\gd$ as in that proof. For $\gd>\bar\gd$, if every candidate's posterior odds at a history are at most $\expo^{-K/(1-\gd)}$, the bounds \eqref{eq:lemma2-odds-growth}--\eqref{eq:lemma2-tv} give each response $\gb^i$ best-response loss at most $\ve$ throughout the next $N_\gd:=\lceil L/(1-\gd)\rceil$ periods. At each history, the normal type's mixed action is common to all candidates. Since best-response loss is affine in the short-lived player's mixed action, $\overline\gb=\sum_i p_i\gb^i$ satisfies the same bound. Equation \eqref{eq:robust-supporting-probability} and full-support monitoring make $\gs^1$ optimal against $\overline\gb$ after every public history. Applying \cref{lem:tciub}, with the actual terminal continuation payoffs against $\overline\gb$, gives a continuation payoff against $\overline\gb$ of at most $\overline W^{CI}+\eta$, uniformly over ambiguity sets, robust equilibria, and supporting weights $p$.

Now repeat the remaining proof of \cref{prop:genfalsification}, replacing its likelihood event by the event that the bound holds for all candidates. Apply the condition with $2K$, as in that proof, to absorb the fixed prior odds. On this event, the preceding continuation bound applies. Since the initial payoff against $\overline\gb$ equals $J_\gd(\gs;\mathcal F)$, the decomposition \eqref{eq:genfalsification-payoff-bound} holds with $U(\gs;\gd,F)$ replaced by $J_\gd(\gs;\mathcal F)$ and failure defined as at least one candidate violating the likelihood bound. The discounted weight before $m_{\gd,2K}$ and the failure probability both vanish uniformly over robust equilibria. Taking the supremum over equilibria, then the limit superior as $\gd\ra1$, and finally letting $\eta\searrow0$ proves the finite-candidate version of \cref{prop:genfalsification}.

I next adapt the construction in \cref{thm:fd-collapse}. Fix $\mathcal E=\{G_1,\ldots,G_m\}$ and a finite horizon $T$. Choose $S\geq T$ such that the distributions $(G_i)_S$ of the first $S$ signals are pairwise distinct. Set $(F_i)_S=(G_i)_S$ and use \eqref{eq:full-support-collapse-tail}, with $T=S$, for every candidate thereafter. The candidates remain distinct and agree with their respective $G_i$ through horizon $T$. The same dates and events as in the proof of \cref{thm:fd-collapse} verify the strengthened likelihood condition: the signal with small subjective probability occurs during the common continuation, every initial probability $(G_i)_S(h_S)$ is at most one, and the normal-type likelihood bound is unchanged. Hence the likelihood inequality holds for every candidate on the same event, whose probability tends to one uniformly over normal-type strategies. The finite-candidate version of \cref{prop:genfalsification} gives $\mathcal F\in\mathfrak F_m^{\mathrm{collapse}}$.

For every $T$, the construction gives $\mathcal F\in\mathfrak F_m^{\mathrm{collapse}}$ whose candidates agree with their respective $G_i$ through horizon $T$. Only horizons beyond $T$ therefore contribute to their distance, so $d_H(\mathcal F,\mathcal E)\leq\max_i d(F_i,G_i)\leq\sum_{s>T}2^{-s}=2^{-T}$. Since $T$ can be arbitrarily large and $\mathcal E$ was arbitrary, $\mathfrak F_m^{\mathrm{collapse}}$ is dense. Full-support candidates are obtained by first replacing each $G_i$ by $(1-\gamma)G_i+\gamma\widehat P$, which preserves distinctness, and letting $\gamma\searrow0$ as $T$ increases. For the stated sequence approaching $\widehat P$, give every candidate the first-$n$-signal distribution $\widehat P_n$, assign $m$ distinct full-support distributions to the signal at date $n$, independently of preceding signals, and use \eqref{eq:full-support-collapse-tail} with $T=n+1$ thereafter. The same argument proves collapse, with exactly $m$ full-support candidates agreeing with $\widehat P$ through horizon $n$.

\paragraph{Survival and fragility to subjective inference.} As in the proof of \cref{thm:nowherecontinuity}, set $F_i^\lambda:=\lambda\widehat P+(1-\lambda)G_i$ and $\mathcal F^\lambda:=\{F_1^\lambda,\ldots,F_m^\lambda\}$ for $\lambda\in(0,1)$. These candidates are distinct and have full support, and $d_H(\mathcal F^\lambda,\mathcal E)\ra0$ as $\lambda\searrow0$. Their common mixture weight gives $\KL{\widehat P_T}{(F_i^\lambda)_T}\leq-\log\lambda$ for every $i,T$, and hence $D_\gd(\widehat P\parallel F_i^\lambda)\leq-(1-\gd)\log\lambda$ by \eqref{eq:discounted-chain-rule}.

Fix $w<R$ and let $\widehat\zeta$ prescribe $\widehat\ga$ after every history. If $w\leq\underline u_{\widehat\ga}$ or $m_{\widehat\ga}(w)=\infty$, this deviation guarantees $w$ against every candidate. Otherwise, put $c_w:=(w-\underline u_{\widehat\ga})/m_{\widehat\ga}(w)$. The payoff-margin argument in the proof of \cref{prop:entropy-rate-robustness} uses only the short-lived players' best-response conditions to bound the deviation payoff. Applied to each candidate, it gives $U((\widehat\zeta,\gs^{2,F_i^\lambda});\gd,F_i^\lambda)\geq w-c_w(1-\gd)\log[1/(\mu_0\lambda)]$. The same deviation is used for every candidate, so taking their minimum and applying robust optimality gives this bound for $\underline W(\gd;\mathcal F^\lambda)$. Letting $\gd\ra1$ and then $w\nearrow R$ proves $\mathcal F^\lambda\in\mathfrak F_m^{\mathrm{canonical}}$. Under $R>\overline W^{CI}$, the density of $\mathfrak F_m^{\mathrm{survival}}$ and nowhere continuity follow exactly as in \cref{thm:nowherecontinuity}.

\paragraph{Fragility to discounting.} Two properties allow the proof of \cref{thm:generic-instability} to apply. First, $\mathfrak F_m$ is Baire. The space of nonempty sets with at most $m$ candidates is compact in the Hausdorff metric, being the continuous image of $\Delta(H)^m$ under $(F_1,\ldots,F_m)\mapsto\{F_1,\ldots,F_m\}$. Its subset with fewer than $m$ candidates is closed, so $\mathfrak F_m$ is an open subset of a compact metric space.

Second, the semicontinuity argument in \cref{lem:payoff-semicontinuity} extends to robust equilibrium payoffs. Near any $\mathcal F=\{F_1,\ldots,F_m\}$, disjoint balls around its members give local labels under which Hausdorff convergence implies convergence of each candidate. The corresponding strategy profiles lie in a common compact product space. Each candidate payoff is jointly continuous in the discount factor and strategies, so their finite minimum is continuous and the normal type's worst-case optimality inequality is preserved for every deviation. At each history, predictions are jointly continuous in the normal strategy and the candidate, by the same positive-denominator argument as in \cref{lem:payoff-semicontinuity}. The robust equilibrium graph is therefore closed. Equilibrium existence and the compactness argument of that lemma give lower semicontinuity of $\underline W$ and upper semicontinuity of $\overline W$, jointly in $(\gd,\mathcal F)$.

The open dense sets in the proof of \cref{thm:generic-instability} can now be defined with $\mathfrak F_m$ in place of $\Delta(H)$ and the robust payoff bounds in place of the ordinary ones. Their openness follows from these semicontinuities, and their density from canonical density and collapse density. The same Baire argument makes $\mathcal G_m$ comeagre and dense. If $R>\overline W^{CI}$, neither $\mathfrak F_m^{\mathrm{collapse}}$ nor $\mathfrak F_m^{\mathrm{survival}}$ intersects $\mathcal G_m$, so both are meagre.
\end{proof}

\subsection{Indeterminacy under reference measures}
\label[onlineappendix]{sec:proof-reference-measures}

Suppose that $R>\overline W^{CI}$. I prove the claim in \cref{sec:generic-instability}: for every $p\in[0,1]$, there exists a nonatomic Borel probability measure $\Pi_p$ on $\Delta(H)$, with full support in the finite-history topology, such that $\Pi_p(\mathcal F^{\mathrm{collapse}})=p$. I construct full-support nonatomic measures concentrated on collapse and on preservation of the canonical guarantee, then mix them.

First, the relevant sets are Borel. Write $V^+(F):=\limsup_{\gd\ra1}\overline W(\gd;F)$ and $V^-(F):=\liminf_{\gd\ra1}\underline W(\gd;F)$. For integers $m>n\geq2$, let $I_{n,m}:=[1-1/n,1-1/m]$. By \cref{lem:payoff-semicontinuity} and compactness of $I_{n,m}$, the map $F\mapsto\max_{\gd\in I_{n,m}}\overline W(\gd;F)$ is upper semicontinuous, and $F\mapsto\min_{\gd\in I_{n,m}}\underline W(\gd;F)$ is lower semicontinuous. Since these intervals increase to $[1-1/n,1)$ as $m\to\infty$, $V^+(F)=\inf_{n\geq2}\sup_{m>n}\max_{\gd\in I_{n,m}}\overline W(\gd;F)$ and $V^-(F)=\sup_{n\geq2}\inf_{m>n}\min_{\gd\in I_{n,m}}\underline W(\gd;F)$. Both maps are therefore Borel measurable, as are $\mathcal F^{\mathrm{collapse}}=\{F:V^+(F)\leq\overline W^{CI}\}$ and $\mathcal F^{\mathrm{canonical}}=\{F:V^-(F)\geq R\}$.

For each $T\geq1$ and $\nu\in\Delta(Y^T)$, let $E_T(\nu)$ be the process whose first $T$ signals have joint distribution $\nu$ and whose subsequent signals follow \eqref{eq:full-support-collapse-tail}. The construction in the proof of \cref{thm:fd-collapse} gives $E_T(\nu)\in\mathcal F^{\mathrm{collapse}}$, with full support whenever $\nu$ has full support. The map $E_T$ is continuous because each finite-horizon marginal is linear in $\nu$, and injective because $E_T(\nu)_T=\nu$. Set $\lambda_T:=1/(T+1)$ and $C_T(\nu):=\lambda_T\widehat P+(1-\lambda_T)E_T(\nu)$. This map is also continuous and injective. By \cref{prop:standard-misspecifications,prop:entropy-rate-robustness}, every $C_T(\nu)$ has full support and belongs to $\mathcal F^{\mathrm{canonical}}$.

Let $\Lambda_T$ be normalized Lebesgue measure on $\Delta(Y^T)$ in its affine hull. Let $\Pi_T^{\mathrm{collapse}}$ and $\Pi_T^{\mathrm{canonical}}$ be the distributions of $E_T(\nu)$ and $C_T(\nu)$, respectively, when $\nu$ has distribution $\Lambda_T$. Define $\Pi^{\mathrm{collapse}}:=\sum_{T=1}^{\infty}2^{-T}\Pi_T^{\mathrm{collapse}}$ and $\Pi^{\mathrm{canonical}}:=\sum_{T=1}^{\infty}2^{-T}\Pi_T^{\mathrm{canonical}}$. Since $|Y|\geq2$, each simplex has positive dimension, so $\Lambda_T$ is nonatomic and assigns probability one to distributions that give every length-$T$ history positive probability. Injectivity of the two maps makes both reference measures nonatomic. They assign probability one to full-support processes in their respective sets.

Both reference measures have full support on $\Delta(H)$. Fix a nonempty open set $O\subseteq\Delta(H)$ and $F\in O$. For every fixed horizon $S$ and every $T\geq S$, $E_T(F_T)_S=F_S$ and $\|C_T(F_T)_S-F_S\|_{\mathrm{TV}}\leq\lambda_T$. Thus $E_T(F_T)\ra F$ and $C_T(F_T)\ra F$. For sufficiently large $T$, both processes belong to $O$. By continuity, $E_T^{-1}(O)$ and $C_T^{-1}(O)$ are nonempty relatively open subsets of $\Delta(Y^T)$, so each has positive $\Lambda_T$-measure. The corresponding terms in the two mixtures therefore give $\Pi^{\mathrm{collapse}}(O)>0$ and $\Pi^{\mathrm{canonical}}(O)>0$.

Finally, set $\Pi_p:=p\Pi^{\mathrm{collapse}}+(1-p)\Pi^{\mathrm{canonical}}$. Under $R>\overline W^{CI}$, the collapse and canonical sets are disjoint. Hence $\Pi_p(\mathcal F^{\mathrm{collapse}})=p$. Both component measures are nonatomic and have full support, so $\Pi_p$ has these properties for every $p\in[0,1]$, including the endpoints.

\end{document}